\documentclass[11pt]{article}

\usepackage{style_file}

\newcommand{\titletext}{Sample correlation adjustments for robust Multi--fidelity Monte Carlo under limited pilot sampling}

\title{\titletext}

\author{
    Michael Stanley\footremember{ama}{Analytical Mechanics Associates, Hampton, VA 23666, USA.} \\ {\small \href{mailto:michael.c.stanley@ama-inc.com}{michael.c.stanley@ama-inc.com}~} \\
    Thomas Coons\footremember{mich}{Department of Mechanical Engineering, University of Michigan, Ann Arbor, MI, 48109} \\
    Geoffrey Bomarito\footremember{nasa}{NASA Langley Research Center, Hampton, VA, 23666, USA.} \\
    Patrick Leser\footrecall{nasa} \\
    Joshua Pribe\footrecall{ama} \\
    James Warner\footrecall{nasa}
}

\date{\today}

\begin{document}

\maketitle

\begin{abstract}
Multi-fidelity Monte Carlo (MFMC) is a variance reduction method that leverages a multi-fidelity ensemble of models of varying cost and accuracy levels. 
Constructing an MFMC estimator with optimal variance requires knowledge of the correlation coefficients between the different fidelity models which are not usually known in practice.
The correlations are typically estimated using offline pilot samples and the sample correlation formula, after which the MFMC method proceeds as if the estimated correlations are the true correlations.
Computational cost often restricts the number of pilot samples used leading to poor correlation estimates and suboptimal estimators.
Leveraging the MFMC problem setting and probabilistic information about the sample covariance matrix, we present a method to improve standard sample-based correlation estimates in the presence of limited pilot samples.
We define a novel discrepancy function quantifying the estimator suboptimality which in turn facilitates selecting a correlation estimator minimizing the worst-case expected discrepancy, where the expectation is taken with respect to the pilot sampling variability.
Through a simple bivariate Gaussian example and a multi-fidelity modeling application from a NASA Entry, Descent, and Landing (EDL) problem, we show that this method produces better MFMC estimators than the standard sample covariance under small pilot sample sizes and limited total budgets.
\end{abstract}

\section{Introduction} \label{sec:intro}

Multi-fidelity (also called multi-model) Monte Carlo methods have revolutionized the Monte Carlo estimation of statistical functionals for forward model uncertainty quantification (UQ) \citep{gorodetsky2020, peherstorfer2016,schaden2020,  warner2021, marten2023}. 
These methods use ideas from the control variates literature \citep{lavenberg1982} to accelerate the estimation of model statistics via variance reduction, optimally constructing an estimator from multiple correlated models with different computational costs. 
Constructing an estimator requires the computational cost of each model, the total computational budget, and the covariance of the model outputs (which, importantly, contains the inter--model correlations).
Although the first two are typically known ahead of time, the covariance matrix often must be estimated in a so-called pilot study.
The pilot study is an offline procedure where each model is run a finite number of times at the same model input values and a sample covariance matrix is computed.
The resulting sample covariance matrix is then treated as the ground truth, and an optimization problem is solved to find the optimal sample allocations and control variate weights.

Importantly, most of the multi-fidelity UQ literature relies on exact knowledge of the covariance matrix through large pilot sample sizes.
In practice, obtaining a large pilot sample size is intractable since the computationally intensive high-fidelity model must be run once per Monte Carlo sample.
As a result, many works \citep{osti_1885882, marten2023, warner2026} have proposed heuristics for pilot sampling termination and more recent works \citep{coons2025bayescov, dixon2026} have investigated methods for determining optimal pilot sample sizes that directly incorporate the covariance-matrix uncertainty and pilot sampling cost into the problem formulation.
Under limited budgets, the pilot sample sizes that one can practically afford are small, often just 5 to 20 pilot samples.
Since the sample covariance matrix is highly variable for these small sample sizes, the downstream estimator hyperparameters are generally suboptimal.

This paper provides a rigorous characterization of the suboptimality resulting from pilot sampling variability, as well as a strategy to address it.
In particular, we propose a framework for producing multi-fidelity estimators that are robust to pilot sampling variability, in turn broadening the practicality of these methods to cost-constrained settings.
To achieve this task, we turn our attention to the sample covariance formula.
While the sample covariance formula has desirable characteristics outside of the multi-fidelity setting, including unbiasedness, it is simply one possible choice for covariance estimation.
Our basic claim is that if one incorporates knowledge of the multi-fidelity estimation problem and the relevant sampling distributions, we can define an alternate covariance estimator that optimally navigates the bias-variance trade-off and can produce better downstream multi-fidelity estimators.

We consider the case where a user has some limited pilot sample data and wishes to produce a multi-fidelity estimator that mitigates the effects of the small pilot sample size.
To produce such an estimator, we first define a \emph{discrepancy function} that quantifies the estimator suboptimality as a result of inexact covariance information.
Next, we define a parameterized \emph{adjustment function} that maps from the sample covariance to a corrected covariance estimate.
Lastly, we formulate and solve an optimization problem to find the risk-minimizing adjustment parameters according to a carefully selected risk metric on the discrepancy function.
Specifically, we find a minimax adjustment on the expected discrepancy, with the expectation taken over the sampling distribution of the sample covariance matrix and the maximum taken over a $1 - \alpha$ confidence set on the true covariance matrix as calculated from the given data.
We call this process the ``data-driven minimax" (DDMM) adjustment, and, to the best of our knowledge, it is the first such method for rigorously handling pilot sampling variability for multi-fidelity forward UQ. 
We also introduce a data-free version of the minimax adjustment but focus our demonstrations on the DDMM procedure.

During our discussion, we intentionally introduce the DDMM adjustment generally ---  the framework can be used regardless of one's choice of adjustment or discrepancy function.
The DDMM method can easily be extended in formulation to settings with multiple low-fidelity models, adjustments that include standard deviation estimates, or other estimators with more complex sample allocation schemes, at the expense of computational cost, which we leave to future work.
We make specific choices in our own implementation that make the problem tractable.
We focus on the bi-fidelity setting and aim to correct the single correlation coefficient between the two available models.
A global sensitivity analysis using Shapley values \citep{owen2014,song2016,owen2017} finds that the correlation coefficient contributes significantly more to the variance of estimator variance than the standard deviations, justifying our choice to only adjust the former.
In addition, we limit our implementation to Multi-Fidelity Monte Carlo (MFMC) \citep{peherstorfer2016} estimators, which have analytical solutions for their variance-minimizing hyperparameter settings and which we show perform well under limited pilot sample sizes in comparison to more general sample allocation schemes.

The contributions of this work are as follows:
\begin{itemize}
	\item We introduce a novel log-ratio discrepancy function quantifying the suboptimality due to pilot sampling covariance variability by directly considering the unrealized reduction in estimator variance without access to exact covariance information. 
	\item We provide an empirical study into estimator robustness and the problem of pilot sampling. We perform a global sensitivity analysis on the MFMC estimator variance using Shapley values, highlighting that the correlation coefficient is the most important covariance component to estimate well. Together with an empirical investigation into the suboptimality of various multi-fidelity estimators beyond MFMC, this result justifies our practical focus on correlation adjustments for MFMC estimators.
	\item Drawing inspiration from Statistical Decision Theory (SDT), we define the DDMM adjustment procedure, which adopts a risk-aware approach to finding an optimal covariance estimate (under our assumptions), and thus estimator hyperparameters, given a set of pilot samples. We then develop a practical methodology for solving the DDMM adjustment optimization problem under the bi-fidelity MFMC setting.
	\item We develop a numerical procedure for determining the optimal DDMM confidence level $\alpha$ such that the adjusted estimator variance is expected to be as small as possible, as quantified by the expected discrepancy function.
	\item We demonstrate the overall method on two examples,
    showing that the method can produce better estimator variances than unadjusted estimators across a variety of estimator budgets and pilot sampling sizes in real multi-fidelity UQ analyses.
	
\end{itemize}

This paper is organized as follows.
In \Cref{sec:suboptimality}, we introduce the problem of pilot sampling by defining the discrepancy function and providing some intuition into how pilot sampling variability produces suboptimal estimators.
Next, in \Cref{sec:robustness}, we empirically explore the pilot sampling problem, providing a comparison of how different popular multi-fidelity estimators fare under limited pilot samples as well as a global sensitivity analysis to elucidate what covariance parameters most drive estimator suboptimality.
Then, in \Cref{sec:DDMM}, we formulate the DDMM procedure as a solution to the suboptimality problem, providing numerical procedures for solving the DDMM optimization and optimally setting the DDMM confidence level $\alpha$.
Lastly, in \Cref{sec:comp_and_emp_results}, we demonstrate the efficacy of the DDMM method on both a numerical benchmark problem and a practical application from the EDL problem from \cite{warner2021} before concluding with future work and final thoughts in \Cref{sec:conclusions}.

\section{Characterizing estimator variance suboptimality}\label{sec:suboptimality}
In this section, we introduce a benchmark problem to elucidate the problem of estimator suboptimality as a result of pilot sample variability.
\Cref{ss:suboptimality} introduces the benchmark problem along with the MFMC estimator and the intuition for suboptimality.
Then, \Cref{sec:discrepancy_function} defines the discrepancy function along with some of its mathematical properties.

\subsection{Demonstration of suboptimality} \label{ss:suboptimality}
The fundamental goal of multi-fidelity methods is to combine different model fidelity output to minimize estimator variance subject to finite computational constraints.
To introduce the sense in which using the sample covariance leads to estimator suboptimality, we use the following toy example.
Suppose we have two computational models, $\cM_i: \mathbb{R}^2 \to \mathbb{R}$ for $i = 0, 1$, where $i = 0$ refers to the high-fidelity (hifi) model and $i = 1$ refers to the low-fidelity (lofi) model.
Each of these model outputs is random due to a random input sampled $\bz \sim \cP$.
Model $\cM_i$ is associated with cost $c_i$ and we assume $c_0 > c_1$.
Given a total compute budget, $C$, we must satisfy the constraint that $n_0 c_0 + n_1 c_1 \leq C$, where $n_0$ is the number of hifi model evaluations and $n_1$ is the number of lofi model evaluations.
In this example, we assume $\cP = \mathcal{N}(\bm{0}, \bI_2)$ and define the following covariance matrix,
\begin{equation}
    \bSigma = \begin{pmatrix} 1 & \rho \\ \rho & 1 \end{pmatrix} = \bL \bL^T,
\end{equation}
with $\bL$ denoting its Cholesky factor.
When we draw a sample $\bz \sim \mathcal{N}(\bm{0}, \bI_2)$ and define $\by = \bL \bz$, we have $\by \sim \mathcal{N}(\bm{0}, \bSigma)$.
We suppose that the first component of $\by$ is the output of our hifi model and the second component the output from our lofi model.
We then have the following:
\begin{equation} \label{eq:toy_dgm}
    y_0 = \cM_0(\bz), \quad y_1 = \cM_1(\bz), \quad \text{Cov}[y_0, y_1] = \rho.
\end{equation}
To facilitate problem articulation, we further assume that we know each of the model output variances (without loss of generality these are taken to be $\text{Var}(\mathcal{M}_i(\bz)) = 1$), and thus we can simply explore the suboptimality of the pilot sample covariance through the sample correlation\footnote{Later, in the global sensitivity analysis study in \Cref{ss:gsa}, we will see that weakening this assumption may not lead to a dramatic shift in the estimator suboptimality.}.
We also assume that we know the model correlation is positive $\rho \in (0, 1)$.
Following the notation of \cite{peherstorfer2016}, the general MFMC estimator of the expectation is defined as,
\begin{equation} \label{eq:estimator}
    \hat{y}(\mfmcweight) := \hat{y}_0 + \mfmcweight \left( \hat{y}_{1+} - \hat{y}_{1-} \right),
\end{equation}
where $\hat{y}_0=n_0^{-1}\sum_i^{n_0}\cM_0(z^{(i)})$ is a hifi Monte Carlo (MC) estimator,  $\hat{y}_{1+}=n_0^{-1}\sum_i^{n_0}\cM_1(z^{(i)})$ is a lofi MC estimator using the same random input set, and $\hat{y}_{1-}=n_1^{-1}\sum_i^{n_1}\cM_1(z^{(i)})$ is a lofi MC estimator using an augmented input set with $(n_1 - n_0)$ additional MC samples.
The MFMC estimator has the following variance under the above model assumptions:
\begin{equation} \label{eq:est_variance}
    \estvar(\params; \rho) := \text{Var}\left[ \hat{y}(\mfmcweight) \right] = \frac{1}{n_0} + \left(\frac{1}{n_0} - \frac{1}{n_1} \right) (\mfmcweight^2 - 2 \mfmcweight \rho).
\end{equation}
We have condensed the left-hand side using $\params = \begin{pmatrix} n_0 & n_1 & \mfmcweight \end{pmatrix}^T$ to indicate the quantities over which the variance is optimized, i.e., the MFMC estimator hyperparameters.
The optimal model evaluation numbers and estimator weight are then found by minimizing the estimator variance subject to the known computational constraints,
\begin{align}
    \underset{\params}{\min} &\quad \estvar(\params; \rho) \label{opt:est_variance} \\
    \text{subject to} &\quad n_0 c_0 + n_1 c_1 \leq C. \nonumber
\end{align}
We refer to the solution to Optimization~\ref{opt:est_variance} by $\params^* = \begin{pmatrix} n_0^* & n_1^* & \mfmcweight^* \end{pmatrix}^T$.
If $c_0 c_1^{-1} > \rho^{-2} - 1$, there is a closed-form solution to the global minimizer of Optimization~\ref{opt:est_variance} \citep{peherstorfer2016}\footnote{Throughout this work, when these conditions are not met, we rely on an isotonic optimization algorithm to find the exact optimal sample allocation. See \Cref{sec:pava} for technical details of the algorithm.}.
Under this particular two-model scenario with unit variances, the solution is:
\begin{equation}\label{eq:optimal_bfmc}
    n_0^* = \frac{C}{c_0 + c_1 r}, \quad n_1^* = n_0^* r, \quad \mfmcweight^* = \rho 
\end{equation}
where
\begin{equation}
    r = \frac{c_0 \rho^2}{c_1 (1 - \rho^2)}. \nonumber
\end{equation}

Suppose we obtain a small number of pilot samples, $\bz_1, \bz_2, \dots, \bz_{\npilot} \sim \mathcal{P}$ (e.g., $\npilot \approx 5$) to estimate the correlation via the sample correlation $\hat{\rho}_{\npilot}$,
\begin{equation} \label{eq:sample_corr_formula}
    \hat{\rho}_{\npilot} := \frac{\sum_{i = 1}^{\npilot}(y^i_0 - \mu_0) (y^i_1 - \mu_1)}{\sqrt{\sum_{i = 1}^{\npilot}(y^i_0 - \mu_0) \sum_{i = 1}^{\npilot}(y^i_1 - \mu_1)}},
\end{equation}
where $\mu_0$ is the hifi model average and $\mu_1$ is the lofi model average.
Please note that in cases where the pilot sample size is obvious, we shorten the sample correlation notation to $\hat{\rho}$.
If we solve Optimization~\ref{opt:est_variance} using the sample correlation, we are minimizing a different objective function, namely $\estvar(\params; \hat{\rho}_{\npilot})$, and hence we obtain a different minimizer, $\hat{\params} = \begin{pmatrix} \hat{n}_0 & \hat{n}_1 & \hat{\mfmcweight} \end{pmatrix}^T$.
It follows that $\estvar(\hat{\params}; \rho) > \estvar(\params^*; \rho)$ almost surely for all $\rho \in (0, 1)$, and thus the estimator variance obtained using the sample correlation is always worse than the true estimator variance obtained under the true correlation.
\Cref{fig:main_idea} provides a graphical illustration for this suboptimality characterization.\footnote{Although \cite{peherstorfer2016} does not guarantee the convexity of the optimization as suggested by our illustration, we chose to represent the variance surface with a convex function to clearly illustrate the idea since there is a closed-form global optimum.}

\begin{figure}[h]
    \centering
    \includegraphics[width=0.75\textwidth]{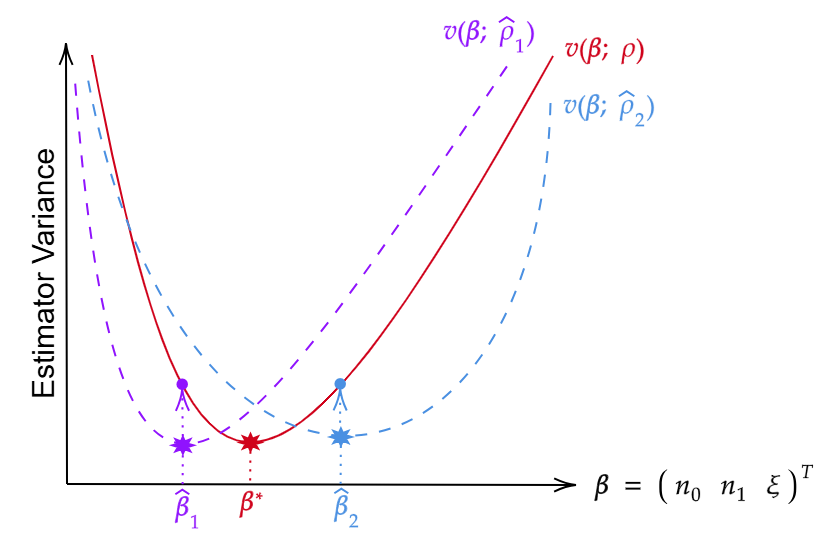}
    \caption{An illustration of the suboptimality challenge this paper addresses. The minimizer resulting from using the sample correlation, $\hat{\rho}_{\npilot}$, is random since it is downstream of its sampling distribution. As such, projecting the optimized $\hat{\params}$ values onto the true estimator variance surface guarantees that the optimized estimator variance is suboptimal with respect to the true correlation.}
    \label{fig:main_idea}
\end{figure}

\subsection{The discrepancy function quantifies suboptimality} \label{sec:discrepancy_function}
Suppose now that we have $\nmodels \in \mathbb{N}$ models and their associated costs.
For a fixed covariance matrix $\bSigma \in \mathcal{S}_{++}^{\nmodels}$ (i.e., the set of symmetric positive definite matrices of dimension $\npilot \times \npilot$), let $\estvar(\cdot; \bSigma): \mathbb{R}^p \to \mathbb{R}$ be the multi-fidelity estimator variance we wish to optimize, where $p$ denotes the estimator hyperparameter dimension, i.e., the total number of sample allocations and weights for the $\nmodels$ models.
We generalize Optimization~\ref{opt:est_variance} as follows:
\begin{align}
    \min_{\params} \quad& v(\params; \bSigma) \label{opt:variance_opt_gen} \\
    \text{subject to} \quad& \bc^T \params \leq C, \nonumber
\end{align}
where $\params =  \begin{pmatrix} \bn & \mfmcweightbold \end{pmatrix}^T$ such that $\bn$ is the sample allocation vector, $\mfmcweightbold$ is the weights vector, $C \in \mathbb{R}_+$ is the total computational budget, and $\bc \in \mathbb{R}^p$ denotes the individual model costs for the first $\nmodels$ components and zeros for the remaining $p - \nmodels$.
Let $\params^*$ denote the global minimizer of Optimization~\ref{opt:variance_opt_gen}.
Suppose we obtain $\npilot \in \mathbb{N}$ pilot samples by generating $\npilot$ I.I.D. samples $\bz_i \sim \mathcal{P}$ (the distribution over the input vectors) for $i = 1, 2, \dots, \npilot$, producing the model outputs $\by_1, \by_2, \dots, \by_{\npilot} \in \mathbb{R}^{\nmodels}$ (henceforth compactly referred to as $\bar{\by}$) with the associated sample covariance matrix $\hat{\bSigma}_{\npilot}$.
Let $\hat{\params}$ denote the global minimizer of Optimization~\ref{opt:variance_opt_gen} but with $\hat{\bSigma}_{\npilot}$ replacing the true covariance.

\vspace{2mm}
Intuitively, we wish to quantify the estimator variance suboptimality resulting from using the pilot samples to estimate the covariance matrix.
We propose the \emph{discrepancy function} as one such option.
\begin{definition}[Discrepancy Function]
    Define the function $\discrepancy: \mathcal{S}_{++}^{\nmodels} \times \mathcal{S}_{++}^{\nmodels} \to \mathbb{R}_+$ as follows,
    \begin{equation} \label{eq:discrepancy_function}
        \discrepancy(\bSigma_1, \bSigma_2) := \log \left( \frac{\estvar(\params_1; \bSigma_2)}{\estvar(\params_2; \bSigma_2)} \right),
    \end{equation} 
    where $\params_1$ minimizes $\estvar(\params; \bSigma_1)$ and $\params_2$ minimizes $\estvar(\params; \bSigma_2)$, both subject to the constraint $\bc^T \params \leq C$.
\end{definition}
\vspace{2mm}

In practice, to evaluate the performance of the sample covariance $\hat{\bSigma}_{\npilot}$ for a fixed true covariance matrix $\bSigma$, we consider $\discrepancy(\hat{\bSigma}_{\npilot}, \bSigma)$.
Clearly, if $\bSigma_1 = \bSigma_2$, then $\discrepancy(\bSigma_1, \bSigma_2) = 0$, otherwise it is positive, reflecting a notion of distance between the obtained estimator variances.
We have refrained from referring to this function as a distance metric since it is not symmetric.
We have also refrained from referring to this function as a loss function despite our inspiration from SDT.
Although our discrepancy function could be characterized as a loss function, we primarily care about the discrepancy of the optimal estimator variances rather than the proximity of $\hat{\bSigma}_{\npilot}$ to $\bSigma$, so we choose this different name to emphasize the distinction.

There are other valid ways to define a discrepancy function such that it quantifies estimator variance suboptimality due to covariance misspecification.
An obvious alternative is a function of the form, $\discrepancy_{regret}(\bSigma_1, \bSigma_2) = \estvar(\params_1; \bSigma_2) - \estvar(\params_2; \bSigma_2)$, where we use the term ``regret'' to indicate its similarity to traditionally defined regret functions (see \cite{berger1985statistical, pred_learn_games}, for example).
Interestingly for MFMC, $\discrepancy_{regret}$ can be seen as a Bregman divergence associated with a particular convex function \cite{pred_learn_games}.
Through this connection, one can potentially make some interesting theoretical insights, but since the discrepancy as defined in \Cref{eq:discrepancy_function} is invariant to budget (see the following paragraph and \Cref{lemma:budget_invariance}), we focus the results in this paper around this choice.
Furthermore, we found our discrepancy function to be more numerically stable than the potential ``regret''-type option.

\subsection{The discrepancy function is invariant to estimator budget}
When using approximate control variate (ACV), which includes MFMC, estimators \cite{gorodetsky2020}, we can make some additional comments about the discrepancy function.
One particular quantity of interest (QoI) for a given multi-fidelity estimator is the \emph{variance reduction ratio} (VRR), quantifying how much the estimator reduces variance relative to the corresponding hifi MC estimator of the same budget. 
Defining the hifi MC estimator variance as $\estvar_0(\bSigma)$, the VRR associated with an estimator with hyperparameters $\params_1$ is defined as $\gamma(\params_1;\bSigma) =\frac{\estvar_0(\bSigma)}{\estvar(\params_1;\bSigma)}$. 

To analyze the behavior of the discrepancy function with respect to the budget $C$, we assume a continuous relaxation of the sample allocations, ignoring nearest-integer rounding constraints.
Under this relaxation, the estimator variance is inversely proportional to the total budget.
We can decompose the hyperparameters into $\params = \begin{pmatrix} C\tilde{\bn}&  \mfmcweightbold \end{pmatrix}^T$, where $\mfmcweightbold$ contains control variate weights and $\tilde{\bn}$ contains the budget-normalized sample allocations satisfying $\bc^T \tilde{\bn} \leq 1$.
The estimator variance then factors as $\estvar(\params; \bSigma) = \frac{1}{C} \tilde{\estvar}(\tilde{\params}; \bSigma)$, where $\tilde{\params} = \begin{pmatrix}  \tilde{\bn} & \mfmcweightbold \end{pmatrix}^T$ represents the budget-independent hyperparameters. 

Because the budget acts as a universal scaling term on both the numerator and the denominator, the VRR is a budget-agnostic measure of the efficiency of a multi-fidelity estimator.
This scaling property naturally extends to our discrepancy function.

\begin{lemma}[Discrepancy function budget invariance] \label{lemma:budget_invariance}
	Under the continuous relaxation of sample allocations, the discrepancy function in Definition~\ref{eq:discrepancy_function} is independent of the total budget $C$ and is equivalent to the log-ratio of the VRRs associated with $\params_1$ and $\params_2$:
	\begin{equation}\label{eq:discrepancy_vrr}
		\discrepancy(\bSigma_1, \bSigma_2) = \log \left( \frac{\gamma(\params_2; \bSigma_2)}{\gamma(\params_1; \bSigma_2)} \right),
	\end{equation}
	where $\params_1$ and $\params_2$ minimize the budget-normalized variances $\tilde{\estvar}(\tilde{\params}; \bSigma_1)$ and $\tilde{\estvar}(\tilde{\params}; \bSigma_2)$ respectively, subject to $\bc^T \tilde{\bn} \leq 1$.
\end{lemma}
\begin{proof}
    By the definition of the variance reduction ratio, $\gamma(\params_2; \bSigma_2) / \gamma(\params_1; \bSigma_2) = \estvar(\params_1; \bSigma_2) / \estvar(\params_2; \bSigma_2)$. Factoring the estimator variances into their budget-dependent and budget-independent components yields:
    \begin{equation*}
        \frac{\estvar(\params_1; \bSigma_2)}{\estvar(\params_2; \bSigma_2)} = \frac{\frac{1}{C} \tilde{\estvar}(\tilde{\params}_1; \bSigma_2)}{\frac{1}{C} \tilde{\estvar}(\tilde{\params}_2; \bSigma_2)} = \frac{\tilde{\estvar}(\tilde{\params}_1; \bSigma_2)}{\tilde{\estvar}(\tilde{\params}_2; \bSigma_2)}.
    \end{equation*}
    The scale factor $1/C$ cancels, demonstrating that the ratio—and consequently the discrepancy function $\discrepancy(\bSigma_1, \bSigma_2)$—is invariant to the total budget constraint $C$.
\end{proof}

\subsection{The expected discrepancy metric}
As described in \Cref{ss:suboptimality}, the discrepancy function becomes a random variable through the sampling distribution of the sample covariance which depends on the distribution of the inputs ($\mathcal{P}$) and the available models.
We assume $\by \sim \mathcal{N}(\bmu, \bSigma)$, i.e., that the model outputs follow a multivariate Gaussian distribution with both unknown mean and covariance.
This assumption is not necessary for the statement of our framework, but it does conveniently lead to the sampling distribution of the sample covariance,
\begin{equation} \label{eq:wishart_sampling}
    (\npilot - 1) \cdot \hat{\bSigma}_{\npilot} \sim \mathcal{W}(\npilot - 1, \bSigma),
\end{equation} 
where $\mathcal{W}(\npilot - 1, \Sigma)$ is a Wishart distribution with $\npilot - 1$ degrees of freedom and scatter matrix $\bSigma$. The associated sample covariance is defined as usual,
\begin{equation}
    \hat{\bSigma}_{\npilot} = \frac{1}{\npilot - 1} \sum_{i = 1}^{\npilot} \left(\by_i - \bmu_{\by} \right) \left( \by_i - \bmu_{\by} \right)^T,
\end{equation}
where $\mu_{\by}$ denotes the sample mean over the vectors $\by_i$.
It is well known in the statistics literature that the Wishart distribution characterizes the sampling distribution of the maximum likelihood estimator of the covariance of a multivariate Gaussian distribution (e.g., \cite{anderson2003}).
This assumption is particularly convenient since the Wishart distribution can be used to estimate statistical functionals of \Cref{eq:discrepancy_function}; we primarily consider the expectation,
\begin{equation} \label{eq:expected_discrepancy}
    \mathbb{E}_{\hat{\bSigma}_{\npilot}} \left[ \discrepancy \left(\hat{\bSigma}_{\npilot}, \bSigma \right) \right].
\end{equation}
As we show in \Cref{sec:DDMM}, we use this notion of expected discrepancy to choose a different estimator of the covariance matrix, as it provides a quantification of a multi-fidelity estimator's (e.g., MFMC) expected suboptimality with respect to pilot sample variability.

\section{An empirical investigation into the robustness of multi-fidelity estimators}\label{sec:robustness}

Now that we have introduced the problem of pilot sampling and defined a discrepancy metric to quantify the suboptimality associated with pilot sampling, we will empirically investigate how multi-fidelity estimators perform in terms of pilot sampling robustness.
Specifically, we will compare a variety of popular multi-fidelity sampling-based estimators across pilot sampling sizes and modeling scenarios to see if there are specific estimators that tend to be more or less susceptible to pilot-sampling variability.
The estimators we compare here include:
\begin{itemize}
	\item Two highly constrained sample allocation strategies: MFMC \citep{peherstorfer2016} and weighted recursive difference (WRDIFF) (an extension on Multilevel Monte Carlo (MLMC) \citep{Giles_2015} that uses optimal weights),
	\item Generalized ACV estimators: approximate control variance independence samples (ACVIS) \citep{gorodetsky2020} and the recursive estimators \citep{bomarito2022},
    \begin{itemize}
        \item Generalized multi-fidelity multiple recursion (GMFMR, generalizing MFMC),
        \item Generalized recursive difference multiple recursion (GRDMR, generalizing WRDIFF),
        \item and generalized independent samples multiple recusion (GISMR, generalizing ACVIS),
    \end{itemize}
	\item and Multilevel best linear unbiased estimator (MLBLUE) \citep{schaden2020}, which adopts a different ansatz for constructing a multi-fidelity estimator based on ordinary least squares and model groupings but that can be shown to be a version of ACV with a unique and highly parameterized sample allocation and weights construction (see \citep{gorodetsky2024}).
\end{itemize}

To carry out this test, we perform an ordered-model experiment extending the toy problem from \Cref{ss:suboptimality} to multiple low-fidelity models, $\nmodels=4$.
We sweep over 8 values of $\rho_{01} \in [0.5,0.95]$ then degrade the subsequent cross-correlations by $70\%$ between each lofi model (and rounding to the nearest positive semi-definite matrix when necessary).
The hifi model is given unit cost and lofi costs are reduced by a factor of $10$ between each level, constructing 8 different modeling scenarios with their own oracle covariance matrices and associated model costs.
For each unique modeling scenario and at a range of pilot sample sizes, we draw $20$ sample covariance matrices $(\npilot - 1) \hat{\bSigma} \sim \wishart (\bSigma,\npilot-1)$, construct the associated estimators with hyperparameters set according to each $\hat{\bSigma}$, and evaluate each estimator's true performance under the oracle covariance.
From these data, we average over the $20$ trials and the $8$ modeling scenarios to generate the expected true estimator variances in \Cref{fig:estvars_ordered} and the expected discrepancy values in \Cref{fig:discrepancies_ordered} as a function of $\npilot$. 

\begin{figure}[ht!]
        \centering
        \includegraphics[width=0.75\textwidth]{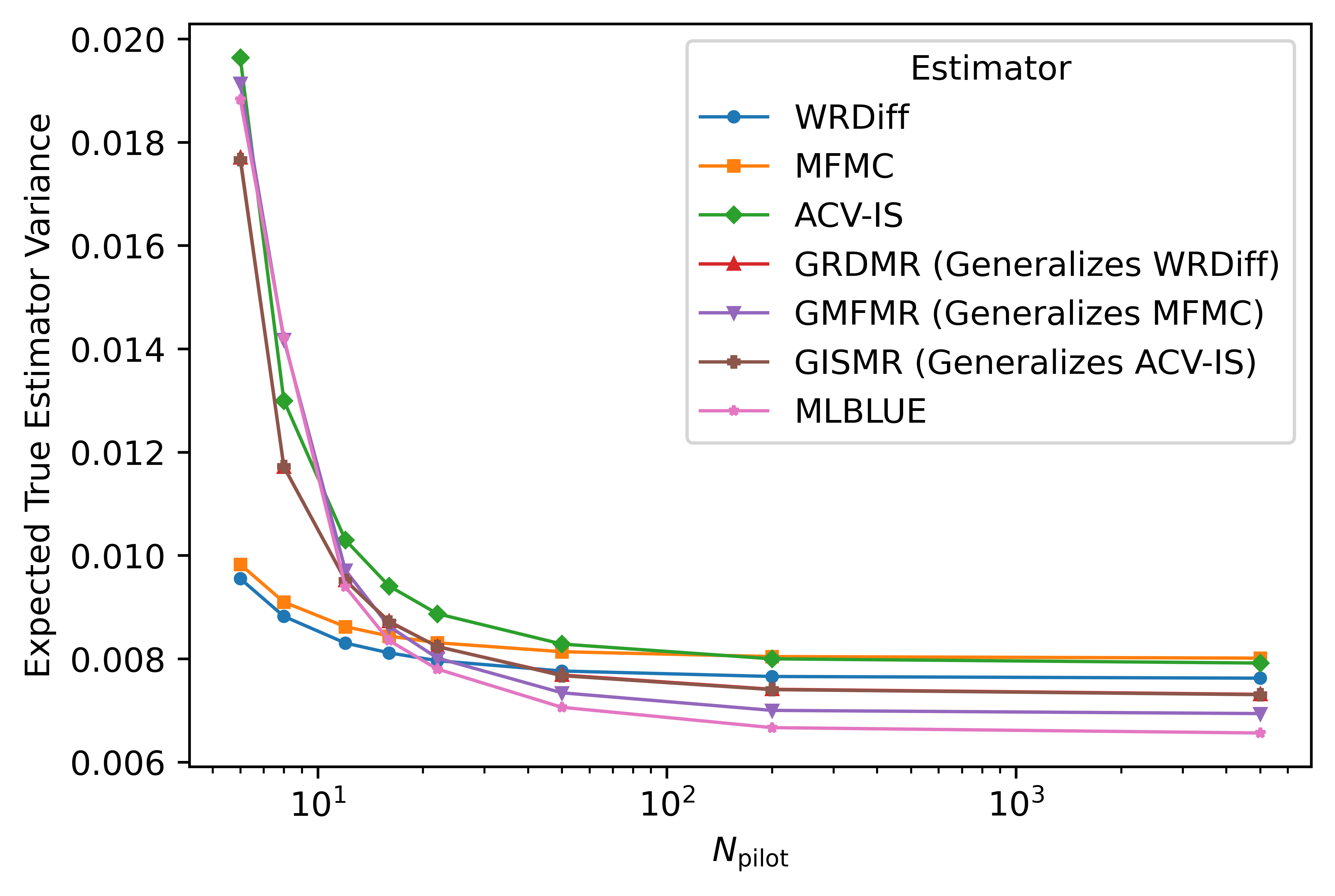}
        \caption{Expected true estimator variances for each multi-fidelity estimator, across different pilot sample sizes. This is a measure of the \emph{absolute} performance of each estimator under limited pilot samples.}
	\label{fig:estvars_ordered}
\end{figure}

The results in \Cref{fig:estvars_ordered} are striking -- while more general estimators such as MLBLUE and generalized ACV outperform when many pilot samples are available, they actually can have far greater estimator variance under smaller pilot sample sizes than the more constrained estimators MFMC and WRDIFF.
There appears to be a clear trade-off between estimator robustness and flexibility in terms of absolute estimator performance under finite pilot samples.
The exact point where MFMC becomes less performant in terms of true estimator variance is problem dependent, but throughout our tests it seems to outperform other methods whenever fewer than $10$ pilot samples are available and there are more than two models.

\begin{figure}[ht!]
        \centering
        \includegraphics[width=0.75\textwidth]{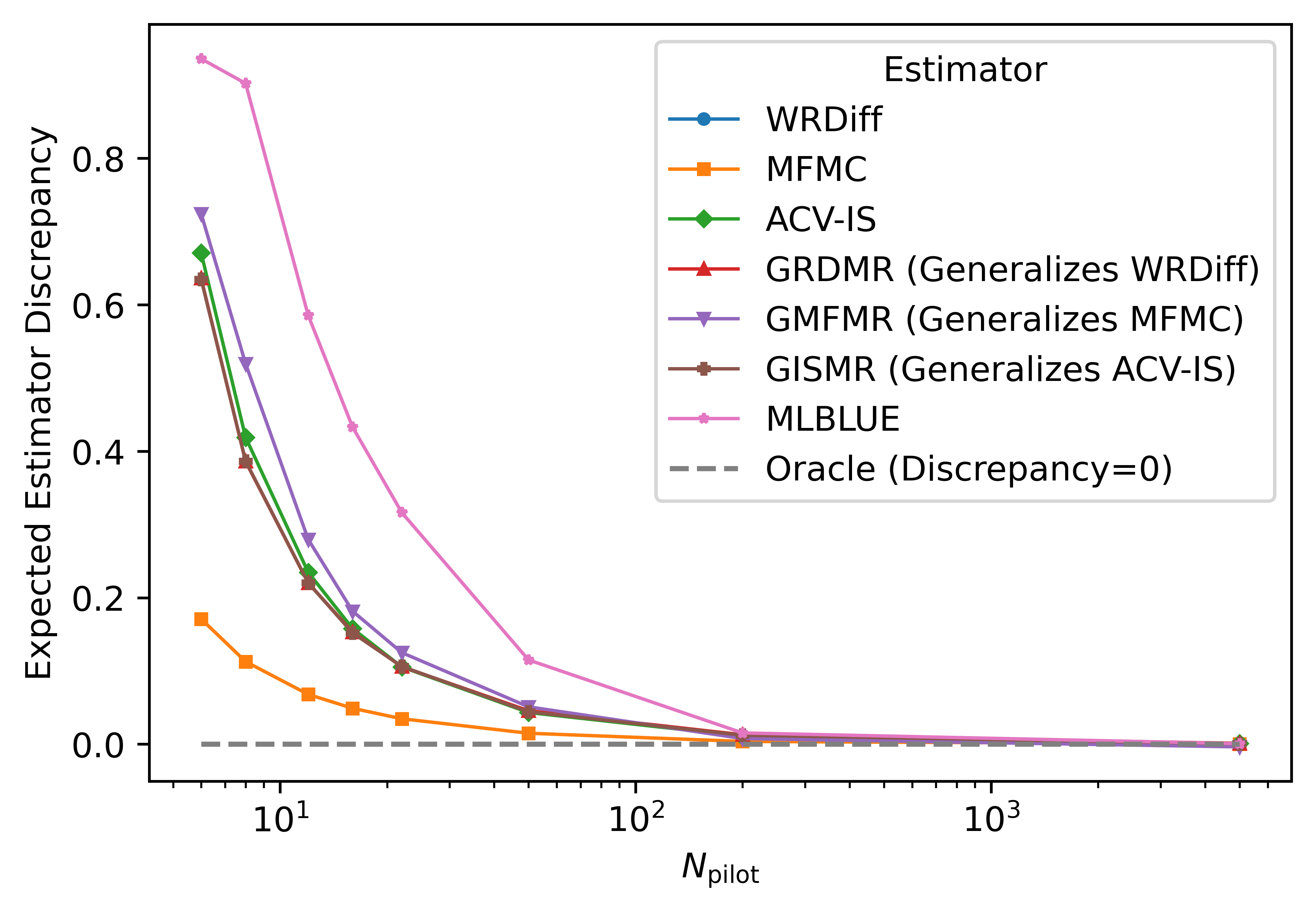}
        \caption{Expected estimator discrepancies for each multi-fidelity estimator, across different pilot sample sizes. This is a measure of the \emph{relative} performance of each estimator under limited pilot samples. The WRDIFF and MFMC curves are completely overlapping in this particular figure.}
	\label{fig:discrepancies_ordered}
\end{figure}

The results in \Cref{fig:discrepancies_ordered} further corroborate this finding. Since each estimator discrepancy is with respect to its own oracle (best-case) estimator variance, this plot shows the relative sensitivity of each estimator to pilot sampling variability.
While more general estimators may eventually overtake MFMC and WRDIFF in terms of absolute performance, they are unequivocally more sensitive to pilot sampling variability in terms of the expected estimator discrepancy. 

We posit two possible explanations for these effects. 
First, the more general estimators have a tendency to ``overfit" their sample allocations to pilot sampling variability.
Since these estimators have fewer constraints on their sample-allocation optimization problem, they make use of every cross-correlation in the provided covariance matrix and are affected by spurious cross-correlation 
pilot-sampling noise.
In contrast, the sample-allocation constraints of MFMC and WRDIFF regularize the associated optimization solutions and make the estimator variance a function of just $\nmodels-1$ cross-correlations, leading to far less variability under the same pilot sampling variability.

Second, calculating the optimal weights can become highly sensitive under these more general estimators due to the linear algebra involved.
Generalized ACV methods require the inversion of a dense $(\nmodels-1)\times(\nmodels-1)$ covariance matrix to find the variance-optimal weights while the corresponding MFMC and WRDIFF matrices are sparse and diagonal.
Since this matrix that must be inverted is a direct function of the provided sample covariance matrix, its inversion can further amplify pilot variability, making MFMC and WRDIFF more robust since inverting a diagonal matrix is more stable.
The advantage of diagonalizing this matrix is especially salient when considering that ACVIS, while it is further generalized by GISMR, does not seem to benefit from the same advantages as MFMC and WRDIFF since it also involves a dense matrix inversion.
Of note, the corresponding matrix inversion for MLBLUE can grow even larger, over all possible models in all possible model groupings, further exacerbating this issue and making it the most sensitive estimator to pilot sampling variability.

A related concern regarding robustness is the \textit{projected} estimator variance, which refers to the estimator variance obtained when evaluating under the sample covariance, $\estvar(\hat{\params};\hat{\bSigma})$.
Since practitioners generally consider this projected variance when deciding how much to trust one's estimator outputs, there can be significant overconfidence issues when it is significantly less than the true estimator error evaluated using the oracle covariance, $\estvar(\hat{\params};\bSigma)$.
Under the above test setup, we plot the ratio between the true estimator variance $\estvar(\hat{\params};\bSigma)$ and the projected estimator variance $\estvar(\hat{\params};\hat{\bSigma})$ across $\npilot$ as a measure of estimator overconfidence in \Cref{fig:ratios_ordered}. We indeed find that the more generalized estimators can underpredict the true estimator variance by an order of magnitude at small pilot sample sizes while MFMC produces far more realistic projected estimator variances even when just $\sim 5$ pilot samples are available.

\begin{figure}[ht!]
        \centering
        \includegraphics[width=0.75\textwidth]{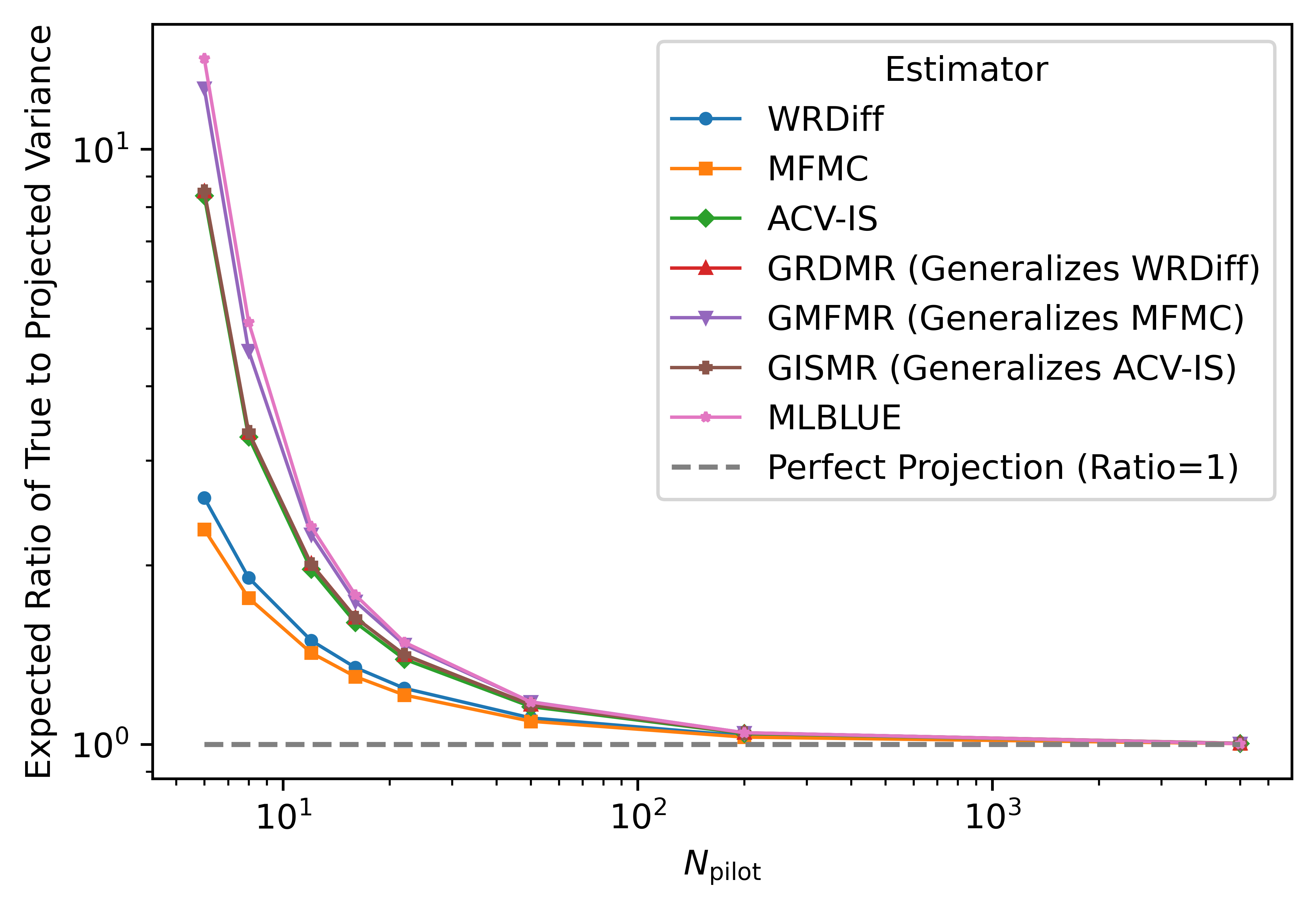}
        \caption{The expected ratio of true estimator variances, $\estvar(\hat{\params};\bSigma)$, to projected estimator variances, $\estvar(\hat{\params};\hat{\bSigma})$ across different pilot sample sizes. This is a measure of the \emph{overconfidence} risk of each estimator under limited pilot samples.}
	\label{fig:ratios_ordered}
\end{figure}

We also repeat a fully random version of this experiment without ordered models (see \Cref{sec:lkj}) and find that the trends above remain under more general model ensembles, albeit with slightly worse absolute performance for MFMC.
Motivated by these results and by the analytical tractability of the optimal MFMC hyperparameters, we limit ourselves to MFMC estimators for the remainder of this work.
However, the formulations herein can be substituted with their counterparts for other estimators, albeit at significantly larger computational burdens, which we leave to future work.

\subsection{A global sensitivity analysis of MFMC suboptimality}\label{ss:gsa}

Global sensitivity analysis (GSA) is a statistical method for evaluating how the uncertainty in model inputs affects the model outputs.
These methods generally assign a sensitivity index to each input of the system, with a larger index indicating a larger impact on the output uncertainty.
In our case, we are interested in how the uncertainty in the sample covariance matrix $\hat{\bSigma}$ affects the uncertainty in the estimator suboptimality $\discrepancy(\hat{\bSigma},\bSigma)$ for some true covariance $\bSigma$.
We use GSA to identify which components of the sample covariance matrix have the greatest impact on the uncertainty in estimator suboptimality, helping guide our proposed correction strategy in \Cref{sec:DDMM}.

We adopt a variance-based approach using Shapley values \citep{owen2017} to properly account for the dependent inputs, namely the sample correlation  $\hat{\rho}$ and sample standard deviations $\hat{\sigma}_0,\hat{\sigma}_1$. 
Rooted in game theory, the Shapley GSA method quantifies each input's fair contribution to the discrepancy metric's variability by averaging each input's marginal contribution across all possible coalitions.
\Cref{sec:shap_app} shows the mathematical formulation and computational details of this approach.

Sweeping over 100 values of $\rho \in [0.01, 0.99]$ for the setup from \Cref{ss:suboptimality}, the results of the Shapley-based GSA are shown in \Cref{fig:shapley}.
Clearly, the main contributor to the uncertainty in the true estimator variance from the sample covariance uncertainty is the correlation coefficient between the low- and high-fidelity model, with an average (over $\rho$) Shapley value of 0.74, while the low- and high-fidelity model standard deviations have an average Shapley value of just 0.09 and 0.17, respectively.
This may be explained by two phenomena.
First, the VRR, and thus the optimal sample allocation, under oracle covariance information is independent of the standard deviations.
This point is elucidated in \citep{marten2023} and therefore removes the standard deviations entirely from the numerator of the discrepancy function, as defined in \Cref{eq:discrepancy_vrr}, for any ACV estimator.
Second, as highlighted in \Cref{eq:optimal_bfmc}, the MFMC weights vector is a function of the \emph{ratio} between each lofi standard deviation and the hifi standard deviation, which may be easier to estimate and more statistically controlled than individual standard deviations when using shared pilot samples.
Motivated by these results, we focus our attention on the correlation coefficient when constructing a protocol to mitigate the effects of pilot sampling variability for MFMC.

\begin{figure}[h]
    \centering
    \includegraphics[width=0.75\textwidth]{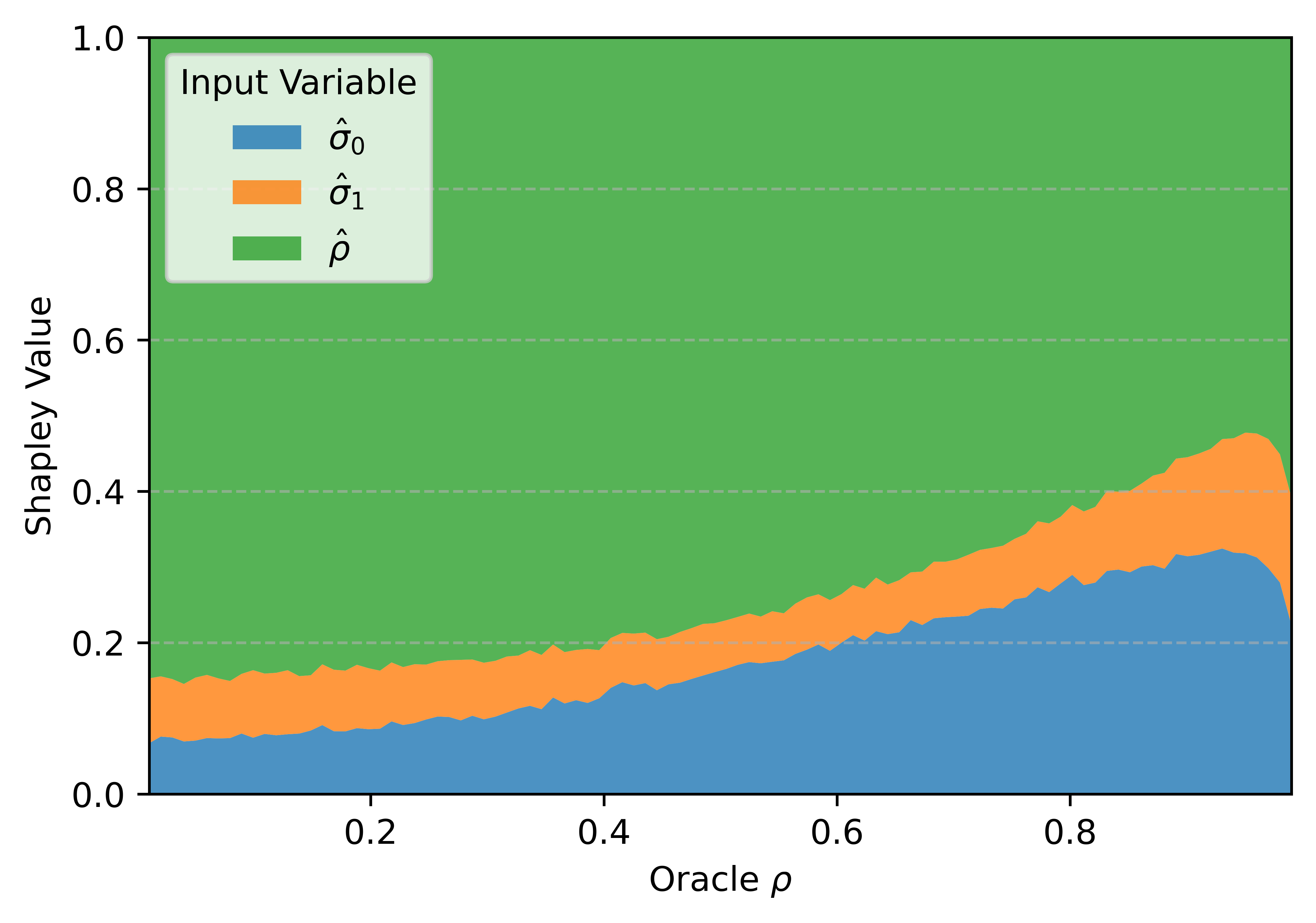}
    \caption{Stacked Shapley values for each component of the sample covariance matrix, for the variance of discrepancy function using the bi-fidelity MFMC estimator. Each $\rho$ corresponds to a different corresponding oracle $\bSigma$, with the oracle model variances each fixed to $1$. The main contributor to the uncertainty in the estimator suboptimality, as quantified by the Shapley values, is clearly the correlation between the high- and low-fidelity model, while the standard deviations are much less impactful.}
    \label{fig:shapley}
\end{figure}

\section{Producing an improved covariance estimator from pilot samples}\label{sec:DDMM}
Using the discrepancy function defined in \Cref{eq:discrepancy_function}, this section provides a procedure to select an improved estimate of the true covariance matrix under pilot sample uncertainty.
Intuitively, the sample covariance computed from the pilot samples provides some information about the true covariance.
Under the Gaussian model outputs assumption, one such piece of information can be probabilistically identified as a $(1 - \alpha)$ confidence set (for $\alpha \in (0, 1)$) in the space of positive definite (PD) matrices.
Leveraging this information, our procedure looks to pick an estimator that minimizes the worst-case expected discrepancy within that $(1 - \alpha)$ confidence set.
The section is structured as follows.
\Cref{sec:evaluation} defines how we characterize estimator optimality using the language of SDT.
\Cref{sec:dd_minimax} then defines a data-driven minimax construction along with key mathematical intuition and theory.
Finally, \Cref{sec:implementation} describes the estimator implementation.
A concise description of the procedure is given by \Cref{alg:ddmm}.

\subsection{Evaluating a covariance estimator} \label{sec:evaluation}
Let $\mathcal{H}$ denote the space of all possible covariance estimators that are a function of observed data such that for each $h \in \mathcal{H}$, $h(\bar{\by})$ produces a PD matrix estimating the true covariance (recall, we use $\bar{\by}$ to compactly refer to the collection of $N$ pilot samples). 
The expected discrepancy function provides a way to compare the performance of different covariance estimators such that the structure of the multi-fidelity scenario is taken into account.
We ideally desire an estimator ($h'$) that is better (i.e., lower expected discrepancy) than the sample covariance for all possible covariance matrices,
\begin{equation} \label{eq:dominating_estimator_sample_covariance}
    \mathbb{E}_{\bar{\by}} \left[ \discrepancy \left(h'\left(\bar{\by} \right), \bSigma \right) \right] \leq \mathbb{E}_{\bar{\by}} \left[ \discrepancy \left( \hat{\bSigma}_{\npilot}, \bSigma \right) \right], \; \; \forall \; \bSigma \in \mathcal{S}^M_{++}.
\end{equation}
In the language of SDT, the existence of $h'$ would make $\hat{\bSigma}_{\npilot}$ \emph{inadmissible} or equivalently, \emph{dominated} by $h'$ (see chapter one of \cite{berger1985statistical}).
More specific to the multi-fidelity problem setting, we wish to find an estimator that is better than the sample covariance for known types of covariance matrices.
For instance, in the scenario explored in the \Cref{ss:suboptimality}, we wish to find an estimator that is better than the sample covariance when $\rho \in (l, u) \subset (0, 1)$.

\subsection{Data-driven minimax (DDMM) construction} \label{sec:dd_minimax}
With the expected discrepancy defined by \Cref{eq:expected_discrepancy}, we define a class of procedures producing a covariance estimator.
We restrict our attention to \emph{adjustments} of the sample covariance matrix and define an adjustment as follows.

\begin{definition}[Adjustment] \label{def:adjustment_function}
    Let $\Theta$ be a space in which we may define a decision (e.g., a Euclidean parameter space like $\mathbb{R}^s$).
    A function $g:\Theta \times \mathcal{S}_{++}^M \to \mathcal{S}_{++}^M$ is considered an adjustment function if it maps the sample covariance matrix to a new covariance matrix.
\end{definition}

For example, let $\btheta \in \Theta \subset \mathbb{R}^s$ and suppose we have an adjustment function $g$ and have computed the sample covariance $\hat{\bSigma}_{\npilot}$ from the $\npilot$ pilot samples.
Then, we write the \emph{adjusted} sample covariance as $g(\btheta, \hat{\bSigma}_{\npilot})$.
We have intentionally limited the scope of \Cref{def:adjustment_function} to use the sample covariance matrix as an input, but this choice is not strictly necessary.
In principle, one could simply define a map from the observed data to a PD matrix.
For example, if we know \emph{a priori} that $\rho \in (0, 1)$, the adjustment function could take the form $\hat{\rho} \mapsto \max\{0, \hat{\rho} \}$.
SDT refers to $\Theta$ as the \emph{action space} and allows us to define a \emph{minimax adjustment} (similar to Definition 1 in Chapter 5V in \cite{berger1985statistical}).

\begin{definition}[Minimax Adjustment] \label{def:minimax_adjustment}
    Under an adjustment function $g$, we call $\btheta^*$ a minimax adjustment if it satisfies the following:
    \begin{equation} \label{eq:minimax_adjustment}
        \sup_{\bSigma \in \mathcal{S}_{++}^{\nmodels}} \mathbb{E}_{\hat{\bSigma}_{\npilot}} \left[ \discrepancy \left(g\left(\btheta^*, \hat{\bSigma}_{\npilot} \right), \bSigma \right) \right] = \inf_{\btheta \in \Theta} \sup_{\bSigma \in \mathcal{S}_{++}^{\nmodels}} \mathbb{E}_{\hat{\bSigma}_{\npilot}} \left[ \discrepancy \left(g\left(\btheta, \hat{\bSigma}_{\npilot} \right), \bSigma \right) \right].
    \end{equation}
\end{definition}

Unlike \cite{berger1985statistical}, we have parameterized the action space through the function $g$ to add structure.
This decision builds off our decision above to define adjustments with respect to the sample covariance.

If one were to compute $\btheta^*$ according to \Cref{def:minimax_adjustment}, it would be maximally robust with respect to the chosen adjustment function.
Furthermore, it would be computable ahead of any observed pilot samples, since its computation only relies upon the chosen variance estimator and statistical assumptions (i.e., Gaussian outputs).
Although this latter property may appear desirable, it could potentially be improved by incorporating the information from the pilot samples as indicated at the beginning of \Cref{sec:DDMM}.
Said differently, this covariance estimator is likely \emph{too} robust and would not render the sample covariance inadmissible as desired.

To address these points, we propose a similar minimax-style estimator that relies upon the pilot samples.
For each fixed adjustment parameter $\btheta$, we propose maximizing over a subset of all possible covariances matrices depending on the sample covariance.
Intuitively, we find a subset of $\mathcal{S}_{++}^{\nmodels}$ containing the true covariance matrix with high probability (e.g., probability $1 - \alpha$ for some $\alpha \in (0, 1)$) resulting in a maximization over a smaller set.
More formally, let $\alpha \in (0, 1)$ and $\hat{\bSigma}_{\npilot}$ be the sample covariance resulting from the pilot samples.
If a function $C_\alpha: \mathcal{S}_{++}^{\nmodels} \to \mathcal{S}_{++}^{\nmodels}$ operates such that
\begin{equation}
    \mathbb{P} \left(\bSigma \in C_\alpha\left( \hat{\bSigma}_{\npilot} \right) \right) \geq 1 - \alpha, \quad \forall \bSigma \in \mathcal{S}_{++}^{\nmodels},
\end{equation}
then $C_\alpha(\cdot)$ is referred to as a $1 - \alpha$ \emph{confidence set} of the true covariance matrix, $\bSigma$.
Note, the output of $C_\alpha(\cdot)$ is a random subset of $\mathcal{S}_{++}^{\nmodels}$ because the pilot samples are random and thus the sample covariance matrix is random.
Additionally, since our empirical results focus on adjustments of the sample correlation rather than the sample covariance, after this subsection we use $C_\alpha(\cdot)$ with the sample correlation ($\hat{\rho}_{\npilot}$) instead of the sample covariance matrix.
This confidence set allows the following modification of \Cref{def:minimax_adjustment}.

\begin{definition}[Data-Driven Minimax (DDMM) Adjustment] \label{def:dd_minimax_adjustment}
    In the context of an adjustment function $g$, we call $\hat{\btheta}(\bar{\by})$ a \emph{data-driven} minimax adjustment if it satisfies the following,
    \begin{equation} \label{eq:dd_minimax_adjustment}
        \sup_{\bSigma \in \hat{C}_\alpha} \mathbb{E}_{\hat{\bSigma}_{\npilot}} \left[ \discrepancy \left(g\left(\hat{\btheta}(\bar{\by}), \hat{\bSigma}_{\npilot} \right), \bSigma \right) \right] = \inf_{\btheta \in \Theta} \sup_{\bSigma \in \hat{C}_\alpha} \mathbb{E}_{\hat{\bSigma}_{\npilot}} \left[ \discrepancy \left(g\left(\btheta, \hat{\bSigma}_{\npilot} \right), \bSigma \right) \right],
    \end{equation}
    where $\hat{C}_\alpha := C_\alpha\left( \hat{\bSigma}_{\npilot} \right)$.
\end{definition}

For all $\btheta \in \Theta$, the worst-case expected discrepancy is probabilistically guaranteed to be lower for the data-driven formulation since the maximization is over a smaller set.
Finding such an adjustment carries the intuition that we do not want to protect against unfavorable covariance scenarios if they are sufficiently far from our sample covariance.
Alternatively, we still wish for our estimator to be probabilistically robust in the sense that the DDMM expected discrepancy is an upper bound on the expected discrepancy at the true covariance matrix. 
The above intuition is formalized in \Cref{prop:prob_bound}.

\begin{proposition}[Inspired by Proposition 1 in \cite{chen2022_dro}] \label{prop:prob_bound}
    Let $\hat{\btheta}(\bar{\by})$ denote the DDMM adjustment specified by \Cref{def:dd_minimax_adjustment}.
    Let $\eta^* := \mathbb{E}_{\hat{\bSigma}_{\npilot}} \left[ \discrepancy \left(g\left(\hat{\btheta}(\bar{\by}), \hat{\bSigma}_{\npilot} \right), \bSigma^* \right) \right]$ for any true underlying covariance matrix, $\bSigma^*$.
    Let $\hat{\eta} := \mathbb{E}_{\hat{\bSigma}_{\npilot}} \left[ \discrepancy \left(g\left(\hat{\btheta}(\bar{\by}), \hat{\bSigma}_{\npilot} \right), \bSigma' \right) \right]$ where $\bSigma'$ maximizes the expected discrepancy over all $\bSigma \in C_\alpha \left(\hat{\bSigma}_{\npilot} \right)$.
    Then,
    \begin{equation}
        \mathbb{P}\left(\hat{\eta} \geq \eta^* \right) \geq 1 - \alpha.
    \end{equation}
\end{proposition}
\begin{proof}
    By definition, the probability that $\bSigma^* \in C_\alpha\left( \hat{\bSigma}_{\npilot} \right)$ is at least $1 - \alpha$.
    If $\bSigma^* \in C_\alpha\left( \hat{\bSigma}_{\npilot} \right)$, it follows that
    \begin{equation}
        \hat{\eta} = \max_{\bSigma \in C_\alpha\left( \hat{\bSigma}_{\npilot} \right)} \mathbb{E}_{\hat{\bSigma}_{\npilot}} \left[ \discrepancy \left(g\left(\hat{\btheta}(\bar{\by}), \hat{\bSigma}_{\npilot} \right), \bSigma \right) \right] \geq \mathbb{E}_{\hat{\bSigma}_{\npilot}} \left[ \discrepancy \left(g\left(\hat{\btheta}(\bar{\by}), \hat{\bSigma}_{\npilot} \right), \bSigma^* \right) \right] = \eta^*. \nonumber
    \end{equation}
    Therefore,
    \begin{equation}
        \mathbb{P} \left(\hat{\eta} \geq \eta^* \right) \geq \mathbb{P} \left( \bSigma^* \in C_\alpha\left( \hat{\bSigma}_N \right) \right) \geq 1 - \alpha. \nonumber
    \end{equation}
\end{proof}

\paragraph{Some intuition for \boldsymbol{$\alpha$}} Since $\alpha$ sets the miscoverage level of the confidence set for the true covariance, it controls an intuitive tradeoff involving robustness and optimality.
When $\alpha$ is small, the confidence set is large and thus DDMM produces an adjustment accounting for a large set of covariance matrices.
In turn, the smallest worst-case expected discrepancy is larger.
When $\alpha$ is large, the confidence set is small and DDMM produces an adjustment acccounting for a relatively small set of possible covariance matrices.
Since there is a tradeoff, it is reasonable to think that there is a \emph{best} $\alpha$ to choose when using DDMM.
Indeed, \Cref{ss:optimizing_alpha} describes our approach for selecting $\alpha$ that provides optimal performance relative to the sample covariance (correlation).

\begin{algorithm}[H]
    \caption{DDMM}
    \label{alg:ddmm}
    \begin{algorithmic}[1] 
    \REQUIRE  $\npilot$ pilot samples, $\by_1, \by_2, \dots, \by_N$. A specified miscoverage level for the confidence set, $\alpha \in (0, 1)$
    \ENSURE An adjusted sample covariance, $g\left(\hat{\btheta}; \hat{\bSigma}_N\right)$.

    \STATE Compute $\hat{\bSigma}_N$ using the pilot samples and compute a $1 - \alpha$ confidence set for $\bSigma$, $C_\alpha\left( \hat{\bSigma}_N \right)$.
    \STATE Compute $\hat{\btheta}$ solving the minimax optimization specified in \Cref{eq:dd_minimax_adjustment}.
    \STATE Use the obtained $\hat{\btheta}$ to compute the adjusted sample covariance, $g\left(\hat{\btheta}; \hat{\bSigma}_N\right)$. 
    \end{algorithmic}
\end{algorithm}

\subsection{Implementing DDMM} \label{sec:implementation}
We focus our implementation on the scenario outlined in \Cref{ss:suboptimality}.
Namely, we consider a single hifi model and single lofi model with known unit variances and thus the DDMM procedure only needs to adjust the sample correlation.

\subsubsection{A sigmoid adjustment function}
To fit our \emph{a priori} knowledge that $\rho \in (0, 1)$, we use a sigmoid adjustment function defined as follows:
\begin{equation}
    g(\btheta; \hat{\rho}) := \frac{\exp \left(\theta_0 \hat{\rho} + \theta_1 \right)}{1 + \exp \left(\theta_0 \hat{\rho} + \theta_1 \right)},
\end{equation}
where $\theta_0$ and $\theta_1$ are the first and second components of $\btheta \in \mathbb{R}^2$, respectively.
By construction, the sigmoid adjustment function ensures that any sample correlation $\hat{\rho} \in (-1, 1)$ gets mapped to $(0, 1)$.

\subsubsection{Computing the confidence interval} \label{ss:computing_cis}
Once the $\npilot$ pilot samples have been observed, we compute $\hat{\rho}$ according to \Cref{eq:sample_corr_formula} and compute $C_\alpha(\hat{\rho})$ using the following method.
$\alpha \in (0, 1)$ can be chosen either heuristically (as is often done in statistics and robust optimization) or computationally (see \Cref{ss:optimizing_alpha}).
Unlike the general exposition in \Cref{sec:dd_minimax}, we compute a confidence interval such that $C_\alpha(\hat{\rho}) \subset (-1, 1)$.

We seek an interval $[\rho_l(\hat{\rho}), \rho_u(\hat{\rho})]$ such that $\mathbb{P}(\rho \in [\rho_l(\hat{\rho}), \rho_u(\hat{\rho})]) \geq 1 - \alpha$ for any $\rho \in (-1, 1)$.
In the following exposition, we eliminate the $\hat{\rho}$ from the interval endpoint notation as it should be clear that these are random endpoints downstream of the random sample correlation.
Equivalently, we want endpoints such that
\begin{equation}
    \mathbb{P} \left(\rho_l \leq \rho \leq \rho_u \right) \geq 1 - \alpha, \; \; \rho \in (-1, 1).
\end{equation}
Since we know the sampling distribution of $\hat{\rho}$ both in closed-form (e.g., \cite{hotelling1953}) or by sampling the appropriate Wishart distribution, finding these endpoints can be framed as a root-finding procedure.
Namely, let $r \in (-1, 1)$ denote the realized sample correlation of the random variable, $\hat{\rho}$.
The lower interval endpoint is the smallest $\rho$ such that the probability that $\hat{\rho} > r$ is equal to $\alpha / 2$.
Formally, $\rho_l$ is the correlation value such that,
\begin{equation} \label{eq:ci_lep}
    \mathbb{P} \left(\hat{\rho} > r \mid \rho_l, \npilot \right) = \frac{\alpha}{2},
\end{equation}
where the above probability notation is the probability of the event $\{\hat{\rho} > r \}$ under the sampling distribution defined by the true correlation set to $\rho_l$ and the number of pilot samples set to $\npilot$.
Similarly, the upper endpoint is the largest $\rho_u$ such that,
\begin{equation} \label{eq:ci_uep}
    \mathbb{P} \left(\hat{\rho} < r \mid \rho_u, \npilot \right) = \frac{\alpha}{2}.
\end{equation}
The root-finding procedure to find these two endpoints simply adjust $\rho_l$ and $\rho_u$ until the desired probability is obtained.

Additional details on this procedure and others can be found in \cite{anderson2003, krishnamoorthy2006}.
Note, for the more general case in which a confidence set in $\mathcal{S}_{++}^{\nmodels}$ is desired, inverting a likelihood ratio test provides a clear path forward.
The probabilities in \Cref{eq:ci_lep} and \Cref{eq:ci_uep} can be computed numerically or via Monte Carlo using the methods described in \Cref{app:samp_corr_wishart}.
For the purpose of choosing an optimal $\alpha$ as described in \Cref{ss:optimizing_alpha}, it is advantageous to learn the interval endpoint surfaces over the $(\alpha, r)$ space, since solving the above root-finding problems can be computationally challenging.
Since the sample correlation density is smoothly varying as a function of both the confidence level ($\alpha$) and the true correlation ($\rho$), using a smoothing or interpolating spline to fit these surfaces produces accurate results.
Details of this approach can be found in \Cref{ss:cis_and_surrogates}.

\subsubsection{Solving the minimax optimization} \label{ss:minimax_solve}
Since solving minimax problems is generally challenging\footnote{Solving this problem using gradient-based optimization falls in the category of stochastic optimization and can be approached using stochastic gradient descent (SGD) \cite{nemirovski2009}. Using SGD-based approaches typically requires convex-concave assumptions on the objective to guarantee solution optimality, however there are more recent efforts to relax these assumptions \cite{jin2020, lin2020}.}, we pursue a discretized solution.
Since the expected discrepancy is neither convex in $\btheta$ nor concave in $\rho$, avoiding numerical (non)gradient-based optimization in favor of computational complexity is prudent.
Let $G_{\btheta} := \{\btheta \}_{i = 1}^{n_{\btheta}}$ denote a grid of $\btheta \in \Theta \subset \mathbb{R}^2$ and $G_\rho := \{ \rho_j \}_{j = 1}^{n_\rho}$ denote a grid of $\rho \in (0, 1)$.
Once $C_\alpha(\hat{\rho}) = [\rho_l, \rho_u]$ has been computed, if given a matrix containing the expected discrepancy values corresponding to the $(\btheta, \rho)$ settings with rows indexed over $G_{\btheta}$ and columns indexed over $G_\rho$, then we simply ignore the columns outside of $C_\alpha(\hat{\rho})$. We can then produce the DDMM solution by first maximizing over all remaining rows then minimizing over the remaining one-dimensional array.
Numerically, the discretization error can be arbitrarily minimized via the grid sizes.
This approach relies upon an expected discrepancy array over  $(\btheta, \rho)$, which can be cumbersome to compute at fine resolutions of $\btheta$ and $\rho$.
As such, in \Cref{ss:exp_discp_arr}, we propose an efficient procedure for computing this high-resolution array which involves first computing a low-resolution array, performing a tensor decomposition, using splines to fit the orthogonal basis functions, and up-scaling the matrix by re-composing the tensor at the desired resolutions.

\section{Computational and Empirical Results} \label{sec:comp_and_emp_results}
We present computational and empirical results to build intuition for DDMM’s operation on a sample level, demonstrate its superior performance under small pilot sample regimes, and support the robustness of its superior performance to assumption violations.
Computationally, we use the toy scenario in \Cref{ss:suboptimality} for which all assumptions hold and show that the DDMM-adjusted sample correlation yields a lower expected discrepancy than the typical sample correlation.
Empirically, we use the NASA EDL multi-model Monte Carlo dataset \cite{warner2021} to show that the theoretical performance superiority holds under assumption violations arising from real data.
We further show that the DDMM adjustment improves the expected variance reduction across all QoIs relative to the unadjusted sample correlation.

In line with \Cref{sec:evaluation} and \Cref{eq:dominating_estimator_sample_covariance}, we wish to show that the DDMM-adjusted sample covariance achieves a lower expected discrepancy than the unadjusted sample covariance across all possible true covariance matrices consistent with our knowledge.
To facilitate this exposition, we consider the \emph{Expected Discrepancy Difference} (EDD) and \emph{Expected Discrepancy Difference percentage} (\%EDD) where the latter is computed by dividing EDD by the unadjusted expected discrepancy.

\begin{definition}[Expected Discrepancy Difference (EDD)] \label{def:edd}
    Given a covariance estimator $h \in \mathcal{H}$, $\npilot$ pilot samples $\bar{\by}$, and $\bSigma \in \mathcal{S} \subset \mathcal{S}_{++}^{\nmodels}$, the EDD is defined as follows,
    \begin{equation} \label{eq:edd}
        EDD := \mathbb{E}_{\bar{\by}} \left[ \discrepancy(h(\bar{\by}), \bSigma) - \discrepancy \left(\hat{\bSigma}_{\npilot}, \bSigma \right) \right].
    \end{equation}
\end{definition}

If the DDMM estimator $h$ produces $\Delta(h, \bSigma) \leq 0$ across all $\bSigma \in \mathcal{S}$, then DDMM dominates the sample covariance.
In the following two sections, we consider using a single hifi and single lofi model and only consider covariance matrices via the correlation.
I.e., in \Cref{sec:bivariate_gaussian_results} we assume the individual model variances are known and equal to one, and in \Cref{sec:edl_results} we only adjust the sample correlation but leave the sample variances untouched when evaluating the (\%)EDD.
Both scenarios consider the (\%)EDD over $\rho \in (0, 1)$.

\subsection{Bivariate Gaussian benchmark} \label{sec:bivariate_gaussian_results}

Under the bivariate Gaussian assumption as described in \Cref{ss:suboptimality}, we demonstrate the performance superiority of DDMM over sample correlation via (\%)EDD.
The Gaussian assumption means the density of the sampling distribution of the sample correlation is known exactly \citep{hotelling1953} and thus the EDD can be computed over $\rho \in (0, 1)$ via numerical integration.
Finally, we assume the hifi and lofi model costs are $c_0 = 1$ and $c_1 = 0.1$ with a total budget $C = 100$.

\begin{figure}
    \centering
    \includegraphics[width=0.9\textwidth]{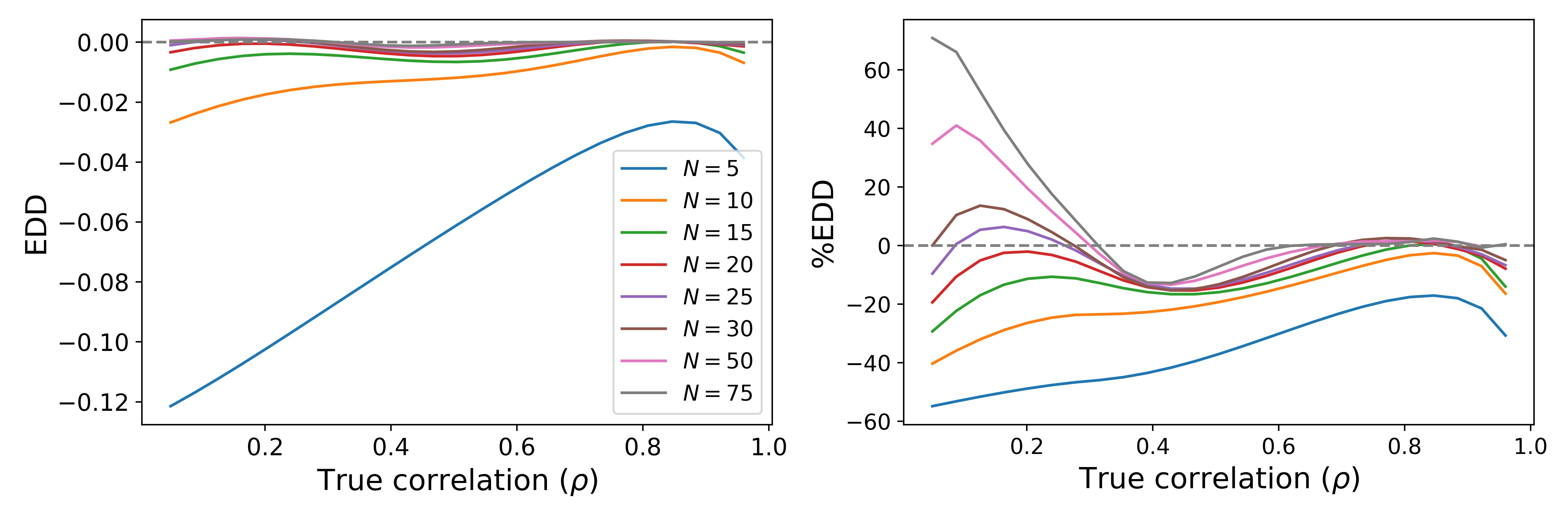}
    \caption{\textbf{(Left)} EDD and \textbf{(Right)} \%EDD across $\rho \in (0, 1)$ the bivariate Gaussian problem, showing DDMM's superior performance over the sample correlation for a collection of pilot sample sizes ($N$). DDMM dominates the sample correlation up to at least $N = 15$.}
    \label{fig:bi_gauss_edd}
\end{figure}

The result of performing this integration across $\rho \in (0, 1)$ is shown in \Cref{fig:bi_gauss_edd}.
The miscoverage level, $\alpha$, is set to $0.253$ as determined by the procedure detailed in \Cref{ss:optimizing_alpha}.
The left panel of \Cref{fig:bi_gauss_edd} shows the EDD across $\rho \in (0, 1)$ for a collection of pilot sample sizes ($\npilot$).
Up to at least $\npilot = 15$ pilot samples, DDMM dominates the sample correlation and is thus worth using in MFMC if one's budget only permits $\leq 15$ pilot samples.
This fact can most clearly be seen in the right panel of \Cref{fig:bi_gauss_edd} showing \%EDD.
Although (\%)EDD is not uniformly below $0$ for $\rho \in (0, 1)$ when $N > 15$, there are still considerable portions of the true correlation space where DDMM maintains an edge.
Thus, even if one can afford more than $\npilot = 15$ pilot samples, depending on what is known about $\rho$, it still may be sensible to use DDMM over the sample correlation.
Additionally, although these results are relevant to this particular example's cost ratio, \Cref{sec:edl_results} explores a scenario where the cost ratio is nearly an order of magnitude smaller and yields similar conclusions.

To summarize the improvement information, \Cref{fig:average_gaussian_improvement} shows the average (\%)EDD values across a range of pilot sample sizes.
The average is taken over the same grid of $\rho$ values producing the plots in \Cref{fig:bi_gauss_edd}.
For pilot sample sizes less than the largest tested pilot sample size ($\npilot = 75$), \%EDD shows that DDMM outperforms the sample correlation on average over $\rho$.
As one might expect, the average improvement diminishes as the pilot sample size grows in large part because the sample correlation distribution becomes much more concentrated around the true correlation.
\begin{figure}
    \centering
    \includegraphics[width=0.75\textwidth]{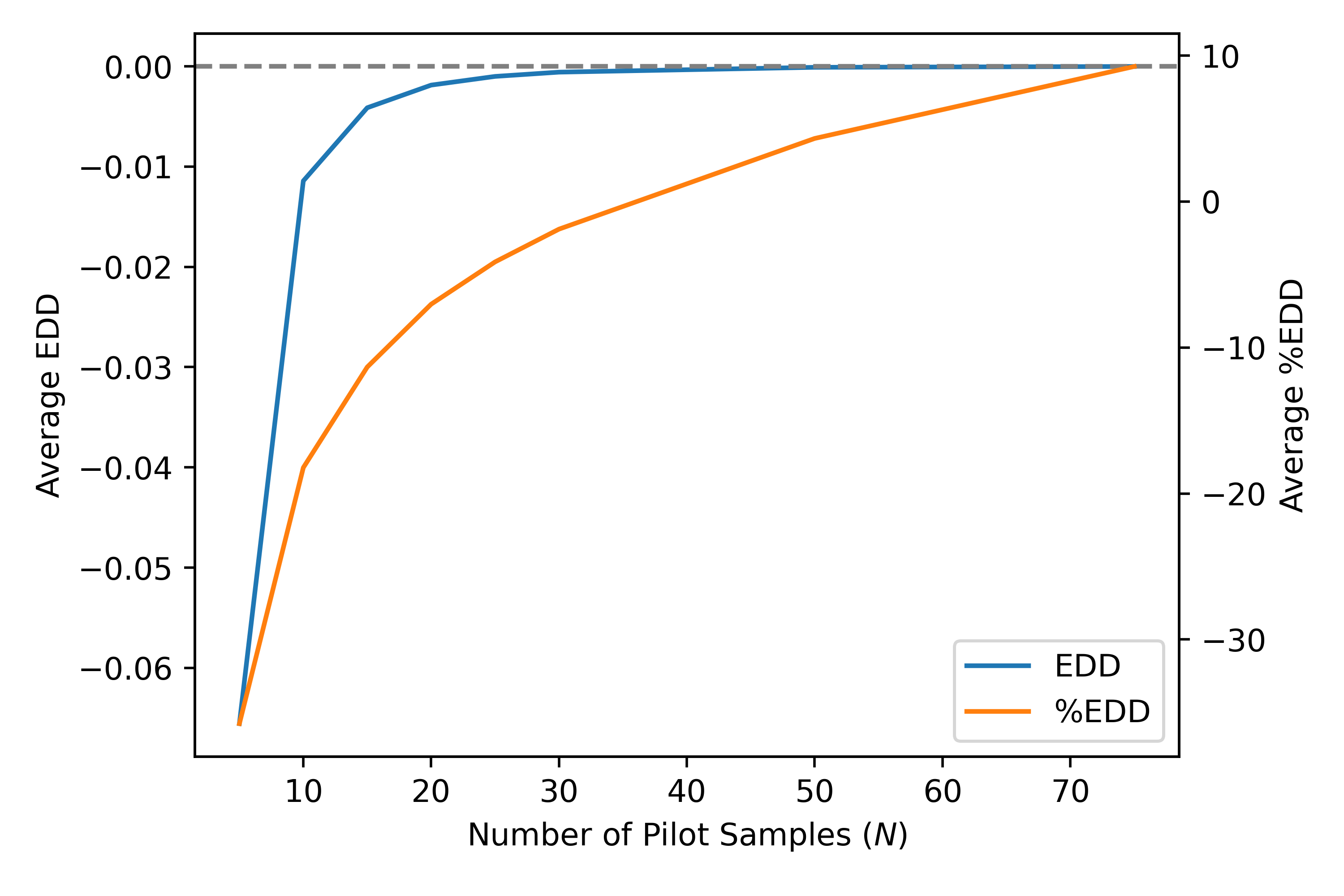}
    \caption{Average (\%)EDD values (over $\rho$) across a range of pilot sample sizes for the bivariate Gaussian problem. For $\npilot < 75$, \%EDD shows that the DDMM procedure procedures an improvement over the sample correlation, on average.}
    \label{fig:average_gaussian_improvement}
\end{figure}

Aside from the Gaussian and covariance assumptions, the above results are downstream of the MFMC configuration (i.e., budget and model costs).
In addition to generating the results with the MFMC budget set to $C = 100$, we generated results for $C = 50$ and $C = 1000$.
These budgets yielded the same results implying that the performance of DDMM is independent to the budget size.
This result is consistent with \Cref{lemma:budget_invariance}.
Thus, the sensibility of including DDMM in an MFMC pipeline is only determined by the absolute number of pilot samples one can afford.
To probe the sensitivity of these results to different model costs, we refer to the results on the EDL application in \Cref{sec:edl_results}.

\subsection{Entry, descent, and landing (EDL) application} \label{sec:edl_results}
The EDL dataset (as explored in \cite{warner2021}) provides a realistic test of multi-model methods and here we use it to evaluate DDMM's performance under a different configuration than that of \Cref{sec:bivariate_gaussian_results} and under violated method assumptions (i.e., non-Gaussian model outputs and unknown model standard deviations).
The data represent trajectory simulations of a sounding rocket and are the output of 75 random inputs, including atmospheric and aerodynamic properties.
The dataset has a hifi model and three lofi model options.
We choose the ``coarse time step'' lofi model option since its cost relative to the hifi model is $0.013$, nearly a factor of $10$ smaller than the relative lofi model cost in the bivariate Gaussian scenario.
Following \citep{warner2021}, we use a time budget of $10^4$ seconds, translating to $45.66$ hifi model runs.
Again, we precompute $\alpha$ to $0.285$ using the procedure detailed in \Cref{ss:optimizing_alpha}.
Although we show global results across all $16$ quantities of interest (QoIs) (see \Cref{fig:edl_variance_reduction}), we focus on Terminal Velocity and Maximum Acceleration (``vel-term'' and ``accel-max'', respectively, in \Cref{fig:edl_variance_reduction}) to showcase the method's operation and sensitivity to assumption violations.
Other QoIs represent features of the trajectory like landing location.
Similar to the bivariate Gaussian scenario, we precompute an optimal $\alpha$ using \Cref{ss:optimizing_alpha}.
Across all QoIs, \Cref{fig:edl_variance_reduction} shows that DDMM produces superior expected variance reduction compared to that of the sample correlation.

\begin{table}[h]
\centering
\renewcommand{\arraystretch}{1.2}
\begin{tabular}{|>{\columncolor{gray!20}\bfseries}l|c|c|c|c|}
    \hline
    \rowcolor{gray!35}
    Metric & Avg. VRR \% Improvement & Avg. EDD & Avg. \%EDD & Avg. MSE \% Chg. \\
    \hline
    Value  & $8.60$\%     & $-0.058$      & $-28.71$\%      & $-7.75$\%      \\
    \hline
    \end{tabular}
    \caption{Summary metrics for the EDL problem, which all show DDMM performance improvement over the sample covariance. The average is taken over the $16$ QoIs considered and shows improvement in the VRR percent improvement (where the percent improvement is calculated by the expected adjusted VRR against the expected unadjusted VRR), improvement in the \%EDD, and improvement in the estimator MSE percent change (where the percent change is the expected adjusted MSE against the expected unadjusted MSE).}
    \label{table:edl_summary}
\end{table}

In addition to the (\%)EDD metrics used in the previous section, we evaluate DDMM using expected mean squared error (MSE) and expected variance reduction where the expectation for both metrics is taken with respect to the randomly drawn pilot samples.
For all metrics, DDMM provides improvments over using just the sample covariance alone.
For each set of $\npilot = 5$ pilot samples, we use DDMM to compute an adjusted sample covariance matrix, solve the MFMC optimization, and evaluate the MFMC variance under the known true covariance (i.e., the covariance matrix computed over all samples).
Since the MFMC estimator is unbiased, this estimator variance under the true covariance is equal to the MSE.
We do the same procedure for the unadjusted covariance matrix.
To sample from each estimator's MSE distribution, we draw $\npilot = 5$ pilot samples $M = 2000$ times from the EDL dataset where the pilot samples are drawn uniformly at random without replacement.
From these samples, we empirically estimate the expectations by averaging over the samples for each of the metrics.
For the variance reduction, we divide the expected variance (MSE) by the MC variance of using the full budget on hifi samples.

\begin{figure}
    \centering
    \includegraphics[width=0.75\textwidth]{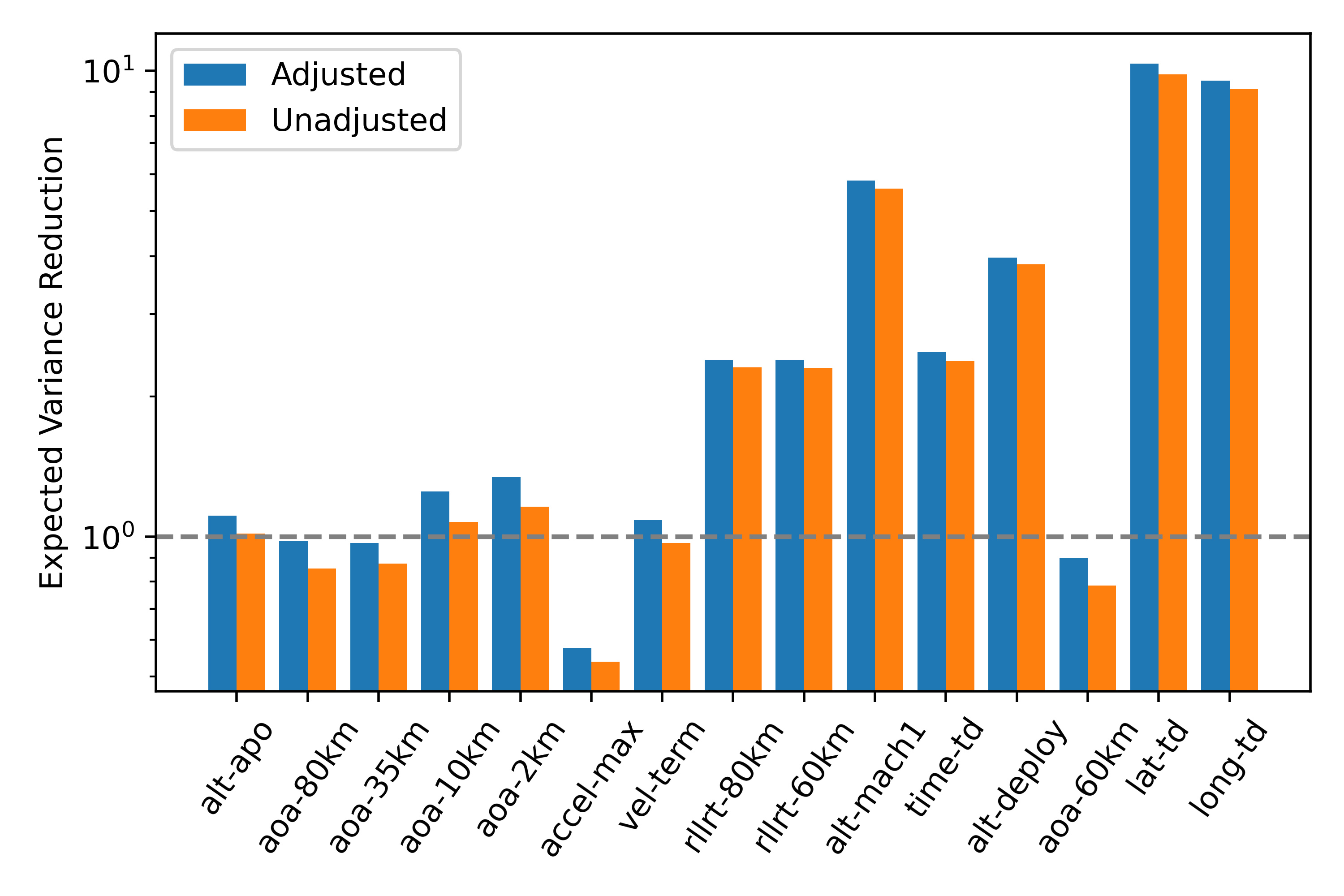}
    \caption{Expected variance reduction for all EDL QoIs where the expectation is respect to the randomness in pilot samples. Across all QoIs, the expected variance reduction is greater when using the adjustment. Additionally, for Terminal Velocity (``vel-term''), adjusting via DDMM is the difference between achieving a variance reduction on average and not.}
    \label{fig:edl_variance_reduction}
\end{figure}

\Cref{fig:edl_edd} shows both theoretical and empirical (\%)EDD values and mean MSE percent changes (from unadjusted to adjusted) across the EDL dataset QoIs.
Like the bivariate Gaussian scenario, across the range of true correlation values in the EDL dataset, we observe DDMM's superior performance via all (\%)EDD values and all mean MSE percent changes falling below zero.
We further observe close agreement between the theoretical values (computed via numerical integration under the bivariate Gaussian assumption) and empirical values (computed by sampling the EDL dataset), indicating that DDMM's benefits are robust to violations in the method's assumptions.
\Cref{table:edl_summary} summarizes the performance improvements observed using the DDMM adjustment versus using just the sample covariance where the average values are taken over the considered QoIs.

\begin{figure}
    \centering
    \includegraphics[width=0.95\textwidth]{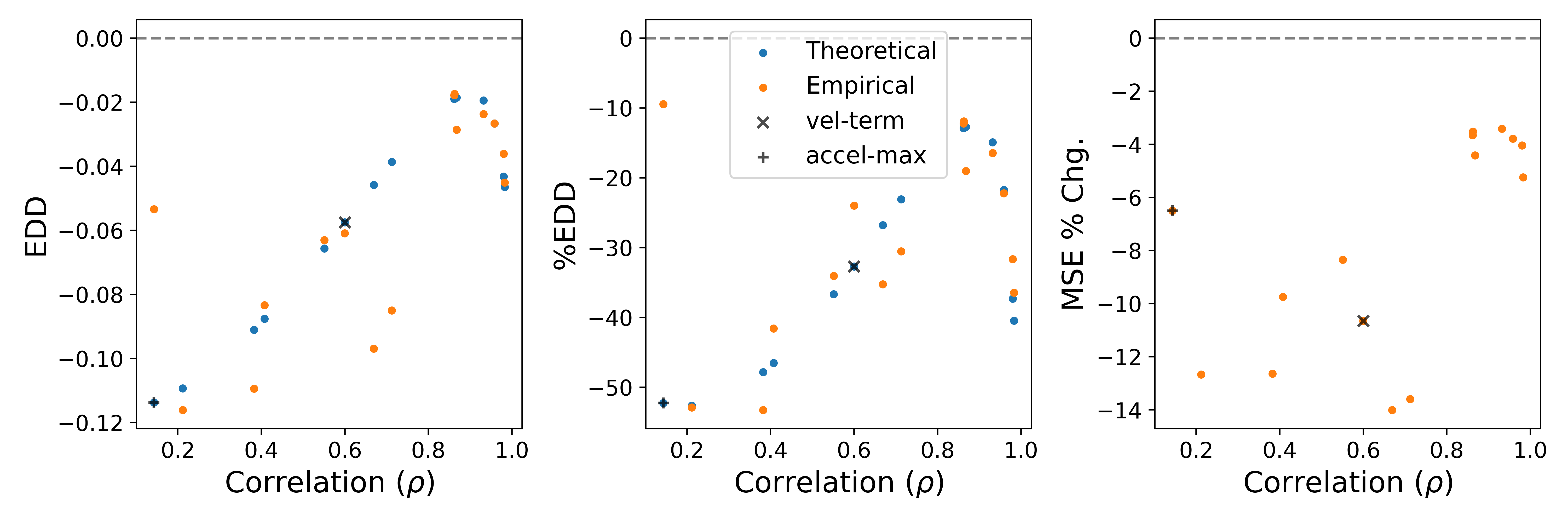}
    \caption{\textbf{(Left)} EDD and \textbf{(Center)} \%EDD across $\rho \in (0, 1)$ show theoretical and empirical values from the EDL data. \textbf{(Right)} similarly shows mean MSE percent changes (unadjusted to adjusted) across the dataset QoIs. Each point corresponds to a QoI. These three plots show that DDMM's superior performance holds across the dataset QoIs with respect to both (\%)EDD and mean MSE metrics. The first two plots further indicate that there is close agreement between the theoretical and empirical values despite assumption violations.}
    \label{fig:edl_edd}
\end{figure}

In all three of the plots in \Cref{fig:edl_edd}, we observe an ``off trend'' QoI at the lowest observed true correlation (the Maximum Acceleration QoI) where the theoretical and empirical values do not match.
Although the scatter plot showing the lofi model outputs against the hifi model outputs clearly shows that the joint distribution is not Gaussian (center panel of \Cref{fig:term_vel_max_acc}), the Gaussian assumption primarily impacts these metrics as it relates to the sampling distribution of the sample correlation (right panel of \Cref{fig:term_vel_max_acc}).
Interestingly, other non-Gaussian QoIs do not seem to have the same problem, highlighting that the sampling distribution of the sample correlations drives the theoretical and empirical mismatch.
For example, the lofi against hifi model outputs for the Terminal Velocity QoI exhibit similarly non-Gaussian properties (left panel of \Cref{fig:term_vel_max_acc}), but the difference between theoretical and empirical EDD is much smaller because the theoretical sampling distribution of the sample correlations more closely matches the empirical sampling distribution (right panel of \Cref{fig:term_vel_max_acc}).
Even under this mismatch, however, the DDMM procedure produces significant empirical performance improvements, albeit at lower levels than theoretically projected.

\begin{figure}
    \centering
    \includegraphics[width=0.95\textwidth]{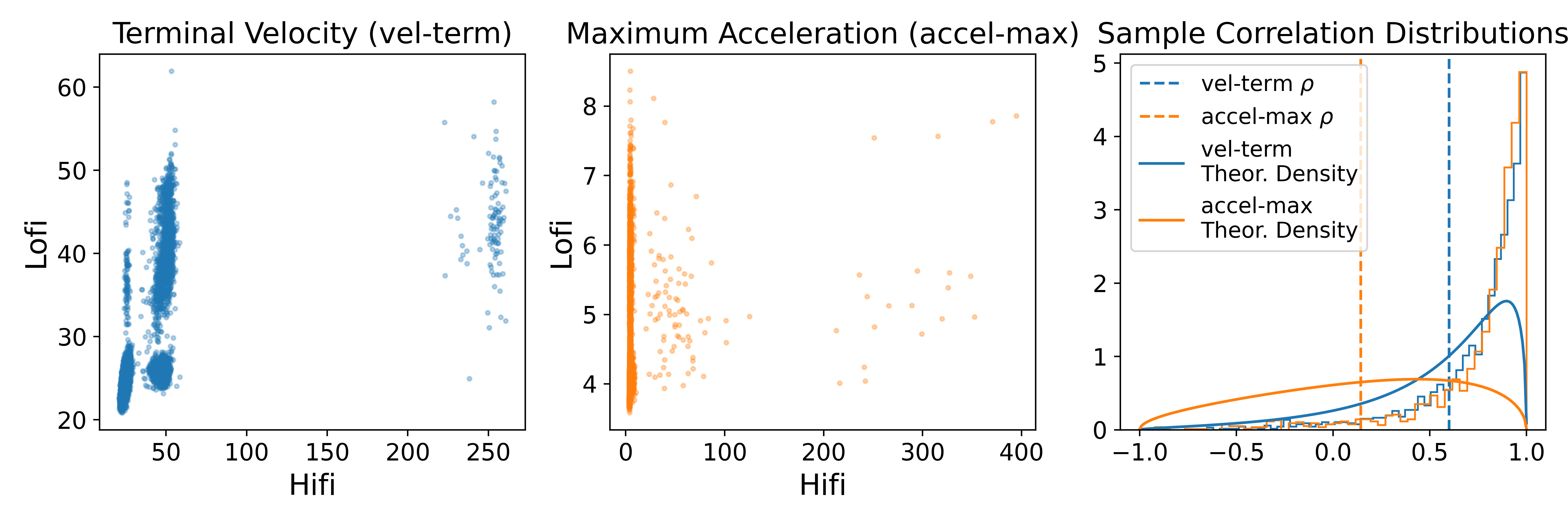}
    \caption{\textbf{(Left)} Lofi versus Hifi outputs for the Terminal Velocity QoI. \textbf{(Center)} Lofi versus Hifi output for the Maximum Acceleration QoI. \textbf{(Right)} Sampling correlation sampling distributions along with their theoretical densities (under the bivariate Gaussian assumption) for both QoIs. Although both Terminal Velocity and Maximum Acceleration are non-Gaussian in a similar qualitative way, their resulting sample correlation distributions differ dramatically with respect to their theoretical densities. The stark distribution difference for Maximum Acceleration is a likely driver for the EDD theoretical/empirical difference observed in \Cref{fig:edl_edd}.}
    \label{fig:term_vel_max_acc}
\end{figure}

Finally, we focus on Terminal Velocity to provide additional intuition into how the method is adjusting the correlation coefficient (see \Cref{fig:term_vel_diagnostics}).
In the left panel of \Cref{fig:term_vel_diagnostics}, we see the adjusted sample correlations plotted against the unadjusted sample correlations, showing that the adjustment is ensuring our positive correlation assumption is enforced by adjusting correlations upward for small or negative sample correlations.
The adjustment also (very slightly) adjusts the correlations downward when $\hat{\rho} \approx 0.5$, and then approximately keeps the sample correlation the same for larger observed values.
In the center panel of \Cref{fig:term_vel_diagnostics}, a quantile-quantile plot of adjusted versus unadjusted MSE values shows that the adjustment is effectively reducing a long right tail in the MSE distribution of the unadjusted procedure, preventing any very large MSEs from propagating to the estimator.
The right panel of \Cref{fig:term_vel_diagnostics} shows that this tail is reduced both for sample correlations substantialy smaller and larger than the true correlation value, highlighting that the procedure is robust to these particularly noisy sample correlations.

\begin{figure}
    \centering
    \includegraphics[width=0.95\textwidth]{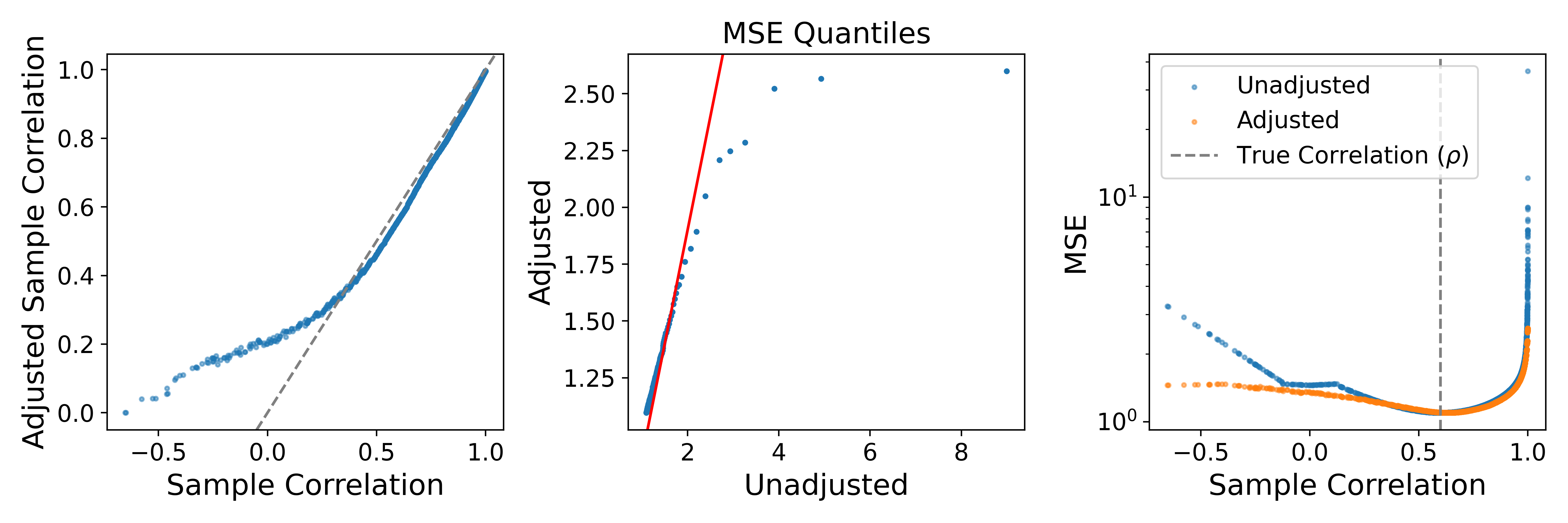}
    \caption{DDMM behavior for the Terminal Velocity QoI. \textbf{(Left)} Adjusted sample correlation versus sample correlation shows how DDMM is adjusting realizations of the sample correlation. \textbf{(Center)} Adjusted MSE quantiles versus unadjusted MSE quantiles shows that the adjustment is reducing a long tail of estimator variance relative to the unadjusted sample correlation. \textbf{(Right)} Realizations of the MSE as a function of the observed sample correlation shows that the adjustment is effectively controlling MSE for sample correlations away from the true correlation.}
    \label{fig:term_vel_diagnostics}
\end{figure}

\section{Summary}\label{sec:conclusions}
For multi-fidelity Monte Carlo procedures, methods are sensitive to pilot sample variability (especially at low sample sizes).
We have proposed a mathematical framework based on a novel discrepancy metric with which this sensitivity can be characterized and quantified, and we have investigated the robustness of multi-fidelity estimators through this lens.
Our empirical investigation into the robustness of these estimators elucidated that there is a clear trade-off between estimator expressivity and estimator robustness, with more generalized estimators suffering both in terms of performance and sensitivity under limited pilot samples while MFMC in particular performs well in these conditions.
We also showed, via a global sensitivity analysis of the MFMC estimator discrepancy, that the uncertainty in the sample correlation coefficient has the greatest impact on estimator variability.
To improve estimator robustness to pilot sample variability, we proposed the data-driven minimax procedure which minimizes the worst-case expected discrepancy over a set of plausible covariance matrices.
Computationally, we demonstrated that the DDMM estimator renders the standard sample covariance inadmissible, in that it produces lower expected discrepancy across all correlation settings $\rho \in (0, 1)$ in the described two-model scenario.
We also observed that the DDMM adjustment uniformly improves performance up to $\npilot = 15$ pilot samples and shows potential gains even for larger pilot sample sizes.
Empirically, using the NASA EDL multi-fidelity Monte Carlo dataset, we showed that DDMM's superior performance holds even when its theoretical assumptions are violated.
Furthermore, the theoretical expected discrepancy difference values were good approximations of the empirical ones obtained for from EDL data for nearly all QoIs in this dataset. 
Finally, we showed that, across all QoIs, the variance reduction was improved when using the DDMM correction over just the pilot sample covariance.

We hope the tools, methods, and results in this paper enable the multi-fidelity Monte Carlo community to better quantify and handle variability due to pilot sampling.
With its current implementation, an optimal $\alpha$ (see \Cref{ss:optimizing_alpha}) and DDMM-adjusted sample correlation can be computed in less than 15 minutes using a personal computer (see \Cref{ss:choose_data_sizes} for the details of the machine used to generate the results of this paper).
As such, adding this adjustment step to an MFMC pipeline has relatively little computational cost while substantially improving estimator robustness.
There are many possible generalizations of this work including different discrepancy functions, different adjustment functions, different statistical notions of estimator optimality, and extensions to multiple lofi model scenarios, which we leave for future work.
Additionally, even though the EDL results show that the improvements can be robust to the Gaussian model output assumption, it would be useful to find ways to avoid the Gaussian assumption altogether. 

\appendix
\section{Sample correlation density and Wishart Distributions} \label{app:samp_corr_wishart}
The ideas in this paper are largely predicated upon our ability to compute expectations of the form,
\begin{equation} \label{eq:exp_corr}
    \mathbb{E}_{\hat{\rho}} \left[f(\hat{\rho}) \right] = \int_{-1}^1 f(r) d P(r),
\end{equation}
for a univariate function $f$, where the expectation is taken with respect to the distribution of the sample correlation.
We also compute probabilities via integration for confidence interval computation as described in \Cref{ss:computing_cis}.
There are few scenarios in which this distribution is known in closed form, one of which is the bivariate Gaussian scenario used throughout this paper.
As such, we have two options for computing \Cref{eq:exp_corr}: numerical integration of the product of the function of interest and the sample correlation density, or Monte Carlo sampling.

To use the former approach, we leverage the density as given in \cite{hotelling1953},
\begin{equation}
    f(r \mid \rho, \npilot) =  \frac{(1 - \rho^2)^{\nu / 2} (1 - r^2)^{\frac{\nu - 2}{2}}}{\mathcal{B}(1/2, \nu/2)} \;_2 F_1 (1/2, 1/2; (\nu + 1) / 2; \rho r),
\end{equation}
where $\mathcal{B}(\cdot, \cdot)$ denotes the Beta function, $\;_2F_1(\cdot, \cdot; \cdot; \cdot)$ denotes the Gauss hypergeometric function, and $\nu = \npilot - 2$.
When $f$ is set to the discrepancy function with a fixed $\btheta$, $\discrepancy(g(\btheta; r), \rho)$, the expectation can be computed using numerical quadrature or Gauss-Legendre quadrature.
The latter happens to be an excellent option in this case since the integral is defined over $[-1, 1]$ and the quadrature can be written in a vectorized summation form for speed.

To use the latter approach, under the bivariate Gaussian assumption we can leverage the Wishart distribution as described by \Cref{eq:wishart_sampling} or slightly more circuitously sample directly from a bivariate Gaussian distribution and compute the sample correlation.
There is a slight advantage to directly using the Wishart distribution, which is that one can sample a realization, $\hat{\bSigma}$, and compute the sample correlation using fewer steps,
\begin{equation}
    \hat{\rho} = \frac{\hat{\Sigma}_{01}}{\sqrt{\hat{\Sigma}_{00} \hat{\Sigma}_{11}}}.
\end{equation}
This shortcut relative to directly sampling from the bivariate Gaussian means that the operation can be easily vectorized, making large sample sizes easy to handle.

Of course, \Cref{eq:exp_corr} is general and does not rely upon any Gaussian assumption.
Given data from some arbitrary distribution, it could be possible to use a subsampling procedure to estimate this expectation as we did in \Cref{sec:edl_results}.
We emphasize that this expectation is on the sampling distribution of the sample correlation which adds a layer of complexity in computing such a value.

\section{Exact MFMC solutions using Pool Adjacent Violators Algorithm (PAVA)}\label{sec:pava}

The analytical solutions to the MFMC estimator hyperparameter problem from \Cref{ss:suboptimality} are only applicable under certain model relationships \cite{peherstorfer2016}.
This can present a significant issue when using non-hierarchical model ensembles or noisy sample correlations since the conditions for these analytical solutions may be violated.
In these cases, users often eschew MFMC for other estimators or rely on gradient-based optimization algorithms to approximately solve the sample allocation problem, increasing the computational burden significantly.
While this computational burden may be small for single estimation tasks where the sample allocation problem only must be solved once, the cost can become prohibitive when evaluating the expected discrepancy metric since the sample allocation problem must be solved many times to numerically estimate the outer expectation.
Here, we introduce an isotonic optimization algorithm called Pool Adjacent Violators Algorithm (PAVA) \cite{busing2022} which provides exact solutions to the (relaxed) sample allocation problem with $\cO(\nmodels-1)$ complexity, alleviating the computational burden of using gradient-based optimization for the MFMC sample allocation problem. 

Suppose we have many ($\nmodels$) computational models, $\cM_i: \mathbb{R}^2 \to \mathbb{R}$ for $i = 0, \dots, \nmodels-1$.
As before, each model $\cM_i$ is associated with a computational cost $c_i$, and we must satisfy the total compute budget constraint $\sum_{i=0}^{\nmodels-1} n_i c_i \leq C$.
As part of the MFMC sample allocation constraints, the sample sets must be strictly nested, enforcing the monotonicity constraint $n_0 \leq n_1 \leq \dots \leq n_{\nmodels-1}$.

The general MFMC estimator of the high-fidelity expectation is defined as:
\begin{equation} \label{eq:gen_estimator}
    \hat{y}(\boldsymbol{\mfmcweight}) := \hat{y}_0 + \sum_{i=1}^{\nmodels-1} \mfmcweight_i \left( \hat{y}_{i+} - \hat{y}_{i-} \right),
\end{equation}
where $\hat{y}_0 = n_0^{-1}\sum_j^{n_0}\cM_0(z^{(j)})$ is the hifi MC estimator, $\hat{y}_{i+} = n_{i-1}^{-1}\sum_j^{n_{i-1}}\cM_i(z^{(j)})$ is the lofi MC estimator using the nested $n_{i-1}$ samples, $\hat{y}_{i-} = n_i^{-1}\sum_j^{n_i}\cM_i(z^{(j)})$ is the lofi MC estimator using the augmented set of $n_i$ samples, and $\boldsymbol{\mfmcweight}$ are the control variate weights.

Defining the model-output covariance matrix as $\bSigma$ and its associated standard deviations as $\sigma_i$ for $i=0,\ldots,\nmodels-1$ and pairwise correlations as $\rho_{ij}$ for $i\neq j$, the variance of the MFMC estimator is given by:
\begin{equation} \label{eq:gen_est_variance}
    \estvar(\params;\bSigma) = \frac{\sigma_0^2}{n_0} + \sum_{i=1}^{\nmodels-1} \left(\frac{1}{n_{i-1}} - \frac{1}{n_i} \right) (\mfmcweight_i^2 \sigma_i^2 - 2 \mfmcweight_i \rho_{0,i} \sigma_0 \sigma_i),
\end{equation}
where $\sigma_i^2$ is the variance of model $\cM_i$, $\rho_{0,i}$ is the correlation coefficient between the hifi model and the $i$-th lofi model, and $\params = \begin{pmatrix} n_0 & \dots & n_{\nmodels-1} & \mfmcweight_1 & \dots & \mfmcweight_{\nmodels-1} \end{pmatrix}^T$.
Setting the partial derivatives with respect to $\mfmcweight_i$ to zero yields the optimal weights $\mfmcweight_i^* = \rho_{0,i} \sigma_0 \sigma_i^{-1}$.
Substituting $\mfmcweight_i^*$ back into \Cref{eq:gen_est_variance} and regrouping the terms by $1/n_i$, the optimal variance simplifies to:
\begin{equation} \label{eq:s_variance}
    \estvar(\params^*;\Sigma) = \sigma_0^2 \sum_{i=0}^{\nmodels-1} \frac{S_i}{n_i},
\end{equation}
where the variance reduction contributions $S_i$ are defined as:
\begin{align}
    S_0 &= 1 - \rho_{0,1}^2 \nonumber \\
    S_i &= \rho_{0,i}^2 - \rho_{0,i+1}^2 \quad \text{for } 0 < i < \nmodels-1 \label{eq:s_terms} \\
    S_{\nmodels-1} &= \rho_{0,\nmodels-1}^2. \nonumber
\end{align}

The sample allocation problem is then defined as:
\begin{align}
    \underset{\mathbf{n}}{\min} &\quad \sum_{i=0}^{\nmodels-1} \frac{S_i}{n_i} \label{opt:isotonic_variance} \\
    \text{subject to} &\quad n_i - n_{i+1} \leq 0 \quad \forall i \in \{0, \dots, \nmodels-1\} \nonumber \\
    &\quad \sum_{i=0}^{\nmodels-1} c_i n_i = C. \nonumber
\end{align}
Following \cite{peherstorfer2016} and applying the notation from \Cref{eq:s_terms}, if the constraints
\begin{align}\label{eqn:mfmc_constraints}
	\frac{c_{i-1}}{c_i} > \frac{S_{i-1}}{S_i}
\end{align}
are met, then the closed-form global minimum for the sample allocation is defined by a ratio vector $\mathbf{r}^*$, where its components for $i = 0, \dots, \nmodels-1$ are:
\begin{equation} \label{eq:optimal_ratios}
    r_i^* =  \sqrt{\frac{c_0 S_i}{c_i S_0}}.
\end{equation}
(Note that by definition, $r_0^* = 1$).
The optimal number of high-fidelity model evaluations ($n_0^*$) is found by distributing the total computational budget $C$ according to these ratios:
\begin{equation} \label{eq:optimal_n0}
    n_0^* = \frac{C}{\sum_{j=0}^{\nmodels-1} c_j r_j^*}.
\end{equation}
Finally, the optimal sample allocations for the lower-fidelity models ($i = 1, \dots, \nmodels-1$) are determined by scaling the high-fidelity sample size:
\begin{equation} \label{eq:optimal_ni}
    n_i^* = n_0^* r_i^*.
\end{equation}
However, when the inequalities in \Cref{eqn:mfmc_constraints} are not satisfied, no such closed-form solution is available.

\subsection{Isotonic Optimization via PAVA}
Due to the strict ordering constraints of the MFMC sample allocation, $n_0 \leq n_1 \leq \ldots \leq n_{\nmodels-1}$, Optimization \ref{opt:isotonic_variance} is a classical isotonic optimization problem.
To solve this problem when the original MFMC constraints from \Cref{eqn:mfmc_constraints} are not satisifed, we utilize the Pool Adjacent Violators Algorithm (PAVA), which is a gold-standard algorithm for isotonic optimization that mathematically satisfies the KKT conditions exactly for convex functions.
PAVA partitions the models into contiguous blocks to enforce the monotonicity constraints.
By tracking the variance-to-cost ratio $S_i / c_i$, the unconstrained allocation dictates $n_i \propto \sqrt{S_i / c_i}$.
Whenever a monotonicity violation occurs (i.e., $R_{i-1} > R_i$), PAVA resolves it by merging the adjacent models into a single pooled block, performing a weighted average of the variance reduction contribution and costs within that block, and continuing to do so until the monotonicity constraints are met.
For the bifidelity case where $\nmodels=2$, this procedure amounts to setting $n_0=n_1=\frac{C}{c_0+c_1}$.

The resulting algorithm is summarized via the following steps:
\begin{itemize}
    \item \textbf{Initialization:} For each model $i = 0, \dots, \nmodels-1$, compute the variance-to-cost ratio $R_i = S_i / c_i$. If not already done so, order the models according to their correlation to the hifi model. Initially, treat each model as its own independent block.
    \item \textbf{Violation Detection:} Evaluate the sequence of blocks for monotonicity violations in $R$. Because the unconstrained optimal sample allocation dictates $n_i \propto \sqrt{R_i}$, a nested sampling violation occurs if a higher-fidelity block requires more samples than its adjacent lower-fidelity block, i.e., $R_{i-1} > R_i$.
    \item \textbf{Pooling:} If a violation is detected, merge the adjacent violating blocks into a single pooled block. The new block is assigned an aggregated variance reduction $S_{\text{pool}} = \sum S_i$ and an aggregated cost $c_{\text{pool}} = \sum c_i$ from its constituent models, yielding an updated pooled ratio $R_{\text{pool}} = S_{\text{pool}} / c_{\text{pool}}$.
    \item \textbf{Iteration:} Continue evaluating and merging adjacent blocks using a stack mechanism until the sequence of ratios across all remaining blocks is strictly monotonically non-decreasing.
    \item \textbf{Sample Allocation:} For each finalized block, determine the pooled sample size $n_{\text{block}} \propto \sqrt{R_{\text{block}}}$. Scale these values uniformly by a single Lagrange multiplier to exactly satisfy the total computational budget $C$, and assign $n_i = n_{\text{block}}$ to all individual models contained within that block.
\end{itemize}
Since PAVA operates in linear time over a stack mechanism, it operates in strictly linear time $\mathcal{O}(\nmodels)$.
We refer the reader to \cite{busing2022, jordan2020} for further details.
In our tests, PAVA can provide orders-of-magnitude speed-ups over gradient-based optimization protocols for modest ($2-30$) model ensemble sizes and benefits from being free from any optimization parameters.

\section{Empirical robustness of multi-fidelity estimators under unordered (random) models} \label{sec:lkj}

We repeat the empirical study from \Cref{sec:robustness}, this time under random model sets. To do so, we replicate the setup from Section 7 of \citep{bomarito2022} with $\nmodels=4$ and some small tweaks to the exact settings.
In this setup, we generate random model scenarios by sampling an oracle correlation matrix from the Lewandowski-Kurowicka-Joe (LKJ) distribution, which can be interpreted as a uniform distribution over valid correlation matrices when its shape parameter is set to 1.
The model variances are generated randomly by setting the hifi variance to 1, then independently sampling the remaining lofi variances from $\mathcal{U}(0.5,1.5)$.
Similarly, the hifi model cost is set to $1\%$ of the total budget while the remaining lofi log-costs ratios are independently sampled from $\mathcal{U}(-4,0)$.
To test each estimator's performance under pilot sampling variability, we once again draw a sample covariance matrix from the associated Wishart distribution, construct the associated estimators with hyperparameters set according to that sample covariance matrix, and evaluate each estimator's true performance under the oracle (LKJ-drawn) covariance.
We repeat this test for 30 trials across a variety of sample sizes.

The true estimator variances are shown in \Cref{fig:lkj_true_vars}, while the expected discrepancies are shown in \Cref{fig:lkj_discrep}.
The results are similar to those under ordered models, but with additional benefit to more general estimators in terms of true estimator variance, especially under large pilot sample sizes.
MFMC and WRDIFF still exhibit superior robustness both in terms of absolute and relative performance, albeit with the more general estimators overtaking in terms of true estimator variance at fewer pilot samples than in the ordered test.
By generating oracle covariance matrices from the LKJ distribution and decoupling the models costs from the model correlations, the model scenarios generated in this test are more amenable to general estimators.
Whereas MFMC and WRDIFF enforce model ordering via their sample allocation constraints, they can produce suboptimal estimators if a very cheap model is far more correlated than an expensive model.
Of note, the LKJ distribution can produce unrealistic model correlations quite often, with negative correlations occurring as often as positive correlations, for example. 
Nonetheless, this test implies that when one has a model ensemble that does not follow a consistent cost/accuracy hierarchy (i.e., with cheap models that may be far better correlated than expensive models), the robustness benefits of MFMC and WRDIFF may not outweigh the expressivity benefits of the other estimators when $\npilot>10$.

\begin{figure}[ht!]
        \centering
        \includegraphics[width=0.75\textwidth]{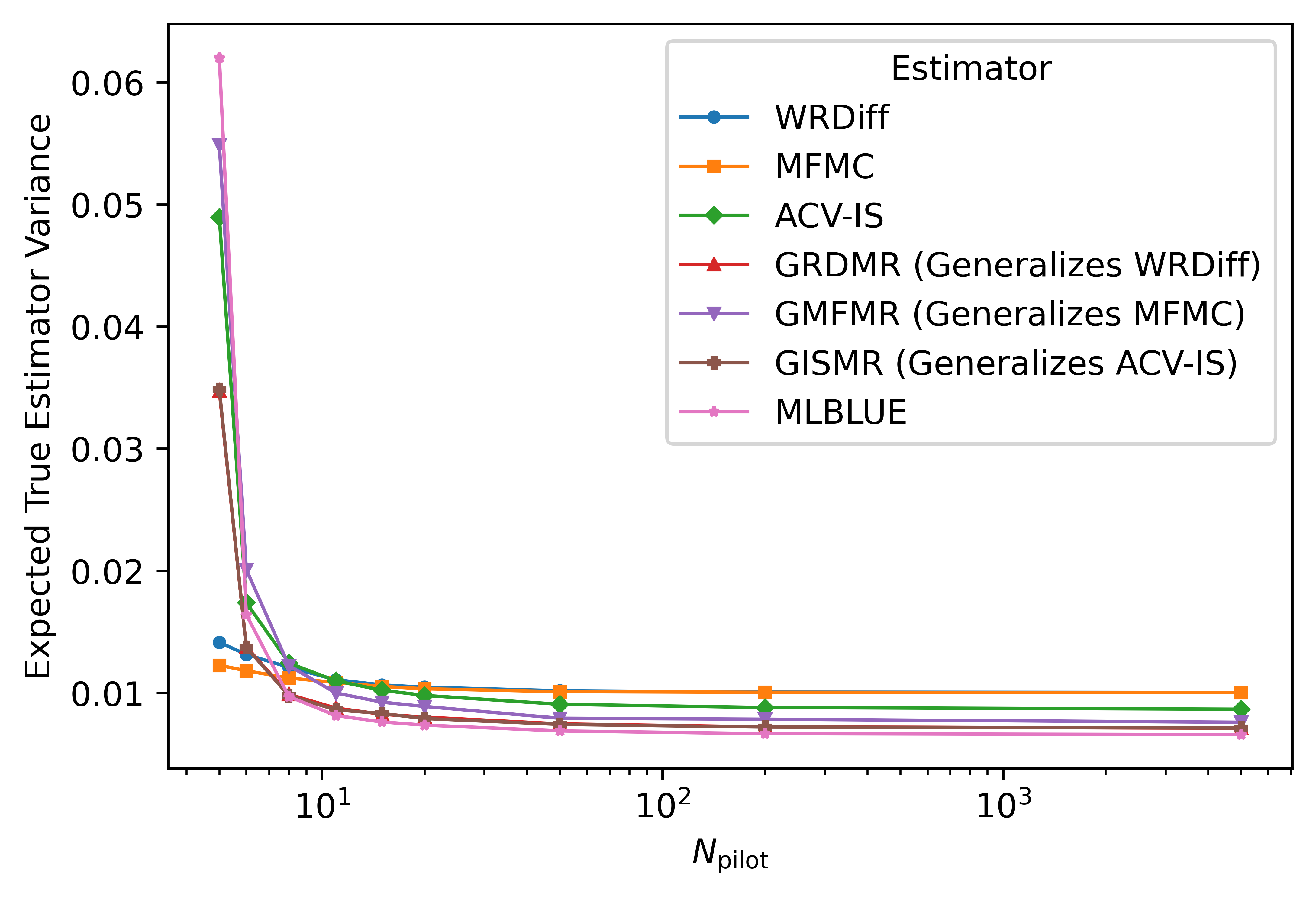}
        \caption{Expected true estimator variances for each multi-fidelity estimator under random model scenarios, across different pilot sample sizes. This is a measure of the \emph{absolute} performance of each estimator under limited pilot samples.}
	\label{fig:lkj_true_vars}
\end{figure}

\begin{figure}[ht!]
        \centering
        \includegraphics[width=0.75\textwidth]{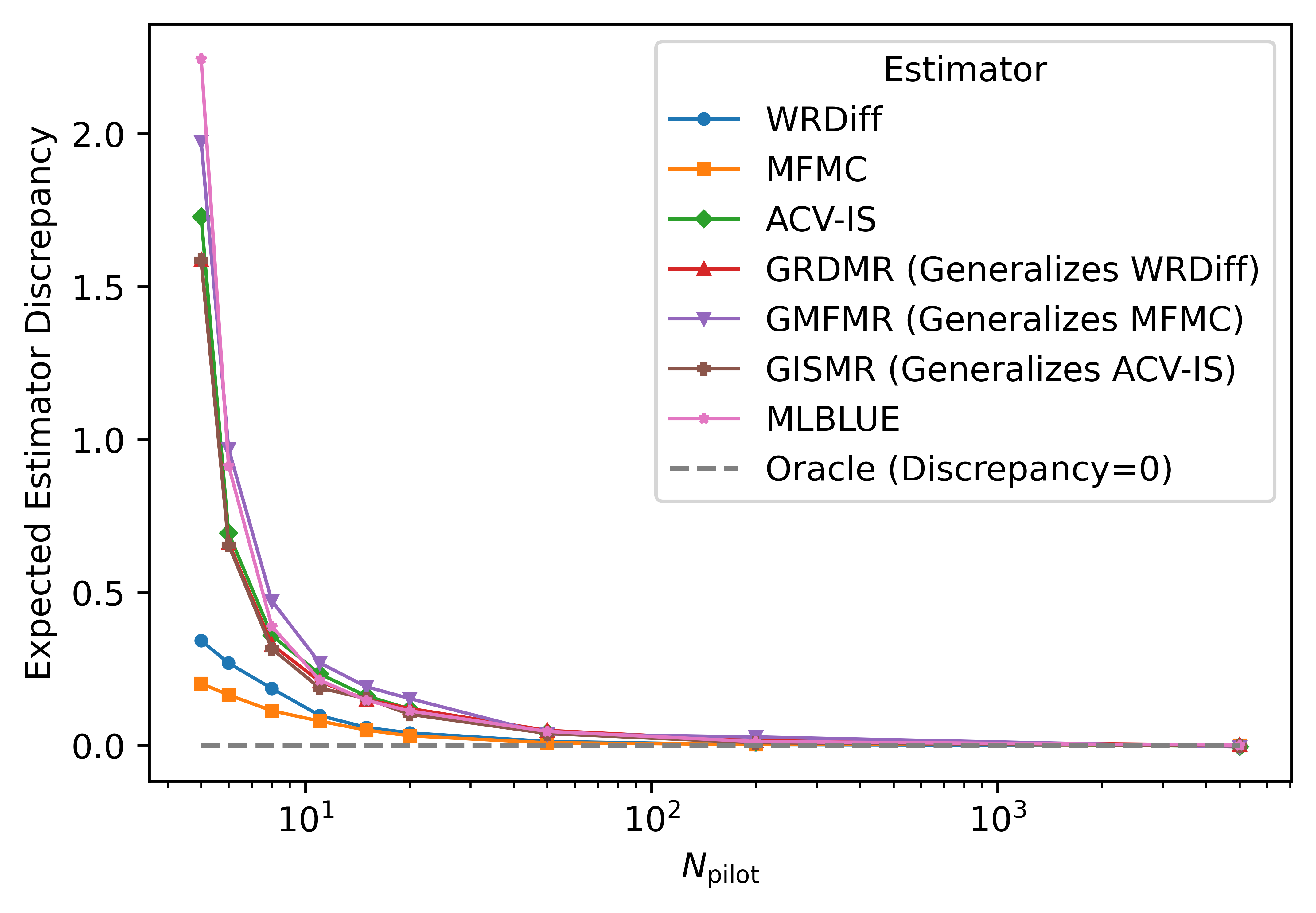}
        \caption{Expected estimator discrepancies for each multi-fidelity estimator under random model scenarios, across different pilot sample sizes. This is a measure of the \emph{relative} performance of each estimator under limited pilot samples.}
	\label{fig:lkj_discrep}
\end{figure}

We also repeat the projected estimator variance test, computing the estimator variance under the sample covariance, $\estvar(\hat{\params},\hat{\bSigma})$ and comparing it to the true estimator variance $\estvar(\hat{\params},\bSigma)$ for each trial, plotted in \Cref{fig:ratios_lkj}.
Interestingly, the problem of false overconfidence is actually worse for the more general estimators in the case of random models.
Since these estimators can better leverage unusual modeling scenarios, the ability to underpredict the true estimator variance under pilot sampling variability appears to be exacerbated, with MLBLUE underpredicting its estimator variance by roughly 23 fold whereas MFMC only underpredicts by a factor of roughly 1.8.

\begin{figure}[ht!]
        \centering
        \includegraphics[width=0.75\textwidth]{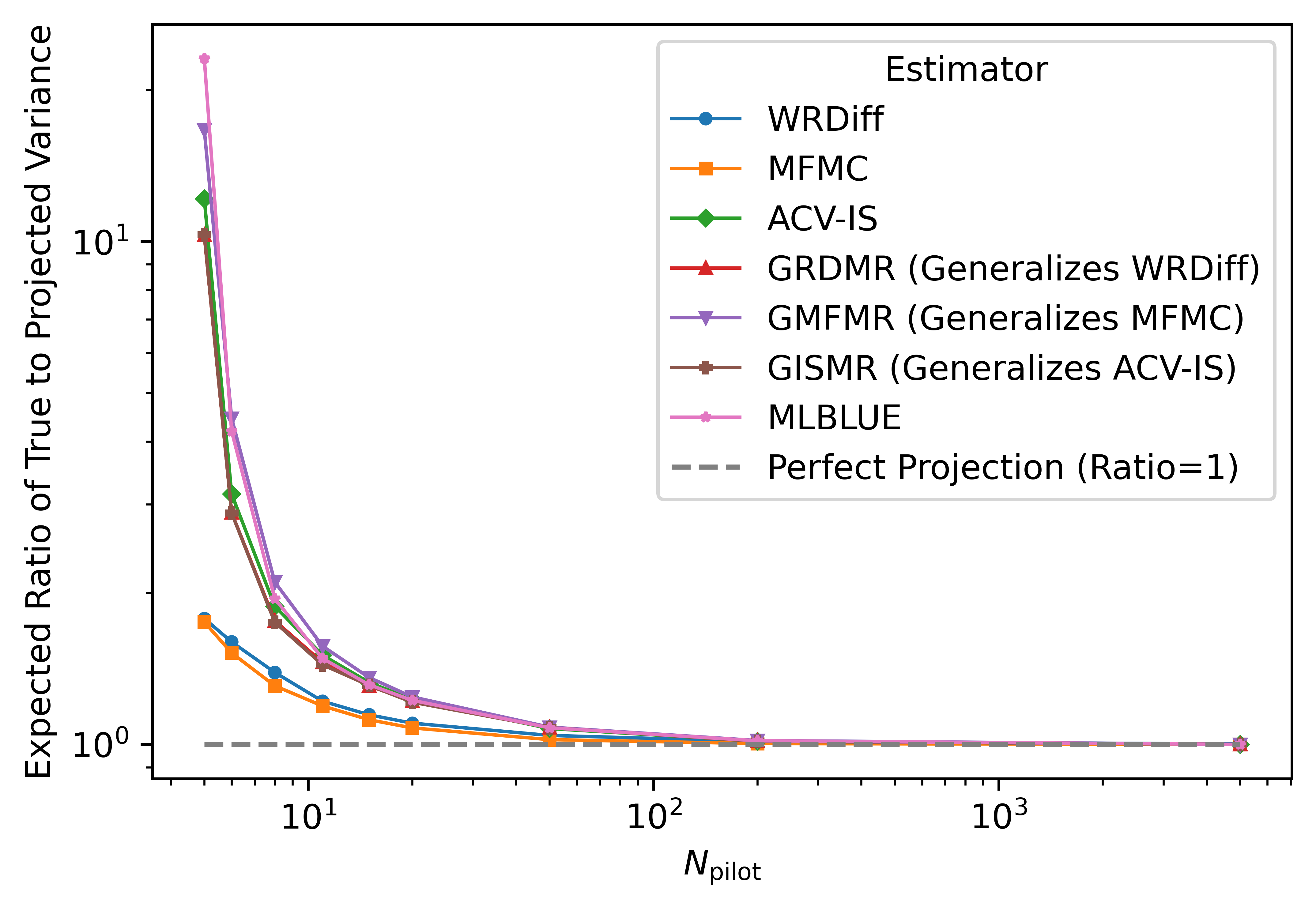}
        \caption{The expected ratio of true estimator variances, $\estvar(\hat{\params};\bSigma)$, to projected estimator variances, $\estvar(\hat{\params};\hat{\bSigma})$ across different pilot sample sizes. This is a measure of the \emph{overconfidence} risk of each estimator under limited pilot samples}
	\label{fig:ratios_lkj}
\end{figure}

\section{Variance-Based Global Sensitivity Analysis (GSA) Using\\Shapley Values}\label{sec:shap_app}

A number of GSA metrics have been proposed in the literature, ranging from variance-based metrics such as Sobol' indices \citep{SOBOL2001271} to density-based metrics such as the $\delta$-importance measure \citep{Borgonovo2017}.
In this work, we restrict ourselves to the bi-fidelity case using MFMC and adopt a variance-based approach based on Shapley values \citep{owen2017} that is amenable to dependent inputs --- since we assume the sample covariance matrix follows a Wishart distribution, $(\npilot-1) \hat{\bSigma} \sim \wishart(\bSigma, \npilot-1)$, it implies a joint distribution over the sample correlation $\hat{\rho}$ and sample standard deviations $\hat{\sigma}_0,\hat{\sigma}_1$ such that these inputs are not independent.
Many methods, such as Sobol' indices, lose their interpretability when the independent inputs assumption is violated.
Many extensions, such as generalized Sobol' indices, aim to disentangle the issue of mutual independence of the model inputs, each with their own interpretations and drawbacks.
The most popular method introduces \textit{Shapley values} (or \textit{Shapley effects}) as a sensitivity measure \citep{owen2014,song2016,owen2017}.
As a quantity used in game theory, Shapley values provide a uniquely fair way to distribute the total ``payoff" (in our case this is the variance of the MFMC estimator variance, $\Var[\discrepancy(\hat{\bSigma},\bSigma)]$) among the ``players" (in our case this is the input parameters, $\mathbf{x}_\cD = \{\hat{\rho},\hat{\sigma}_0,\hat{\sigma}_1\}$). Interestingly, the seminal work  \citep{Shapley} proves that it is the only such measure that satisfies the four axioms of a fair game, namely efficiency, symmetry, dummy/null, and additivity. 

The Shapley value $\phi_d$ for an input $x_d$ is its average marginal contribution to the variance, computed over all possible subsets $S$ of the other inputs. Let $\cD$ be the index set of the inputs, with cardinality $D=\vert\cD\vert$.
The general formula is:
\begin{align}\label{eqn:shapley}
    \phi_d = \frac{1}{\Var[\discrepancy(\hat{\bSigma},\bSigma)]} \sum_{S \subseteq \cD \setminus \{d\}} \frac{|S|! (D - |S| - 1)!}{D!} \times \nonumber  \left( \Var[\mathbb{E}[\discrepancy(\hat{\bSigma},\bSigma) \vert x_{S \cup \{d\}}]] - \Var[\mathbb{E}[\discrepancy(\hat{\bSigma},\bSigma) \vert x_S]] \right)
\end{align}
In this formulation, $S$ is a subset of inputs not containing $d$, and $\Var[\mathbb{E}[\discrepancy(\hat{\bSigma},\bSigma) | x_S]]$ represents the variance of the MFMC discrepancy that is explained by the subset of inputs $x_S$. This formula calculates the ``fair" contribution by averaging the marginal increase in explained variance that $x_d$ provides, weighted across all possible coalition sizes.
Computationally, $\phi_d$ is estimated using MC or kernel methods --- we do so using the nearest neighbor search method proposed in \citep{broto2020,azadkia2021} in the $\mathtt{sensitivity}$ package in $\mathtt{R}$ using $10^4$ MC samples at each tested oracle $\bSigma$.

\section{Technical details for solving DDMM} \label{sec:solving_ddmm_tech_details}

\subsection{Computing the expected discrepancy array} \label{ss:exp_discp_arr}

As stated in \Cref{ss:minimax_solve}, we trade the theoretical complexity of solving the minimax optimization of DDMM for computational complexity such that the solution can be accurately approximated using vectorized array operations.
Implementing this idea requires computing expected discrepancy values over a dense grid of $\btheta \in \Theta$ (the adjustment parameter space) and $\rho \in (0, 1)$ (the true correlation space) to ensure small approximation error.
While such a dense gridding can be achieved in a brute-force embarrassingly parallel approach, we concluded this approach would not realistically be feasible for implemention on a personal computer.
As such, we instead leverage smoothness assumptions about the expected discrepancy \emph{surface} over $(\btheta, \rho)$.

For each $(\btheta, \rho) \in \Theta \times (0, 1)$, we define the expected discrepancy surface as follows,
\begin{equation} \label{eq:M_theta_rho}
    M(\btheta, \rho) := \mathbb{E}_{\hat{\rho}} \left[ \discrepancy\left(g(\btheta; \hat{\rho}), \rho \right) \right] = \int_{-1}^1 \discrepancy\left(g(\btheta; r), \rho \right) f(r \mid \rho, N) dr,
\end{equation}
where $f(r \mid \rho, N)$ is the sample correlation density under the bivariate Gaussian assumption as defined in \Cref{app:samp_corr_wishart} when the true correlation is $\rho$ and the sample correlation is composed of $\npilot$ pilot samples.
As remarked in \Cref{app:samp_corr_wishart}, $f(r \mid \rho, N)$ changes smoothly as a function of $\rho$ and since $g(\btheta; r)$ is a sigmoid function, it changes smoothly as a function of $\btheta$.
Thus, it is reasonable to assume the surface $M(\btheta, \rho)$ is smooth over its inputs.

Since $M$ is smooth and defined over a compact domain, it is square integrable ($M \in L^2$).
Thus, the operator defined by $M$ is a Hilbert-Schmidt operator and implies the kernel has an infinite sum expansion,
\begin{equation}
    M(\btheta, \rho) = \sum_{i = 1}^\infty \sigma_i u_i(\btheta) v_i(\rho),
\end{equation}
where $\sigma_i \geq 0$ and $\sigma_i \to 0$ as $i \to \infty$, where the $\sigma_i$ values rapidly decay \cite{rudin1991functional}.
For each $i$, $u_i(\btheta)$ is also a bivariate function and if we use a similar infinite expansion,
\begin{equation}
    u_i(\btheta) = \sum_{j = 1}^\infty \lambda_{ij} b_{ij}(\theta_0) c_{ij}(\theta_1),
\end{equation}
hence, we get the following tensor decomposition,
\begin{equation}
    M(\btheta, \rho) = \sum_{k = 1}^\infty v_k(\rho) b_k(\theta_0) c_k(\theta_1),
\end{equation}
where we reindex $(i, j)$ to $k$ and absorb the $\sigma_i$ and $\lambda_{ij}$ values into the basis functions.
Since the spectrum rapidly decays, we rely upon the following approximation,
\begin{equation} \label{eq:finite_approx}
    M(\btheta, \rho) \approx \sum_{k = 1}^R v_k(\rho) b_k(\theta_0) c_k(\theta_1),
\end{equation}
where $R < \infty$.

Motivated by \Cref{eq:finite_approx}, \Cref{alg:exp_discrp_compute} computes $M$ at a small collection of $(\btheta, \rho)$ locations within its compact domain to form an array.
We then perform a canonical polyadic (CP) tensor decomposition using the tensorly package \cite{kossaifi2019tensorly}, which  provides us with a finite collection of evaluations of the orthonormal functions $v_k$, $b_k$, and $c_k$.
For each function, we use a cubic spline to approximate the function values at the $(\btheta, \rho)$ values not present in the grid.
Finally, using the fitted splines, any desired grid resolution of the function values can easily be obtained and recomposed by recomposing the tensor.

For the minimax optimizations in this paper, we achieved negligible minimax solution error using the following fine-dimension settings; $d'_{\theta_0} = 200$, $d'_{\theta_1} = 200$, and $d'_\rho= 1000$.
To provide a guideline around the choice of $R$ in \Cref{eq:finite_approx} and demonstrate the accuracy of this computational approach, we computed the output of \Cref{alg:exp_discrp_compute} using a brute force approach (i.e., we parallelized computing each $M_{ijk}$ component) and the approach in \Cref{alg:exp_discrp_compute} using different $R$ values.
Denote the brute-force output by $\bm{M}_{true}$ and the output of \Cref{alg:exp_discrp_compute} by $\bm{M}(R)$.
We consider relative error as a function of $R$,
\begin{equation}
    e(R) := \frac{\lVert \bm{M}(R) - \bm{M}_{true} \rVert_F}{\lVert \bm{M}_{true}\rVert_F},
\end{equation}
and show the result in the left panel of \Cref{fig:exp_discrep_diagnostics}.
The center and right panels of \Cref{fig:exp_discrep_diagnostics} show two arbitrarily selected orthonormal functions from both the brute-force expected discrepancy array and the reconstructed one.
Overall, both closely match, although we observe a slight deviation in the center panel on the boundary of the $\rho$ space.

\begin{figure}
    \centering
    \includegraphics[width=0.95\textwidth]{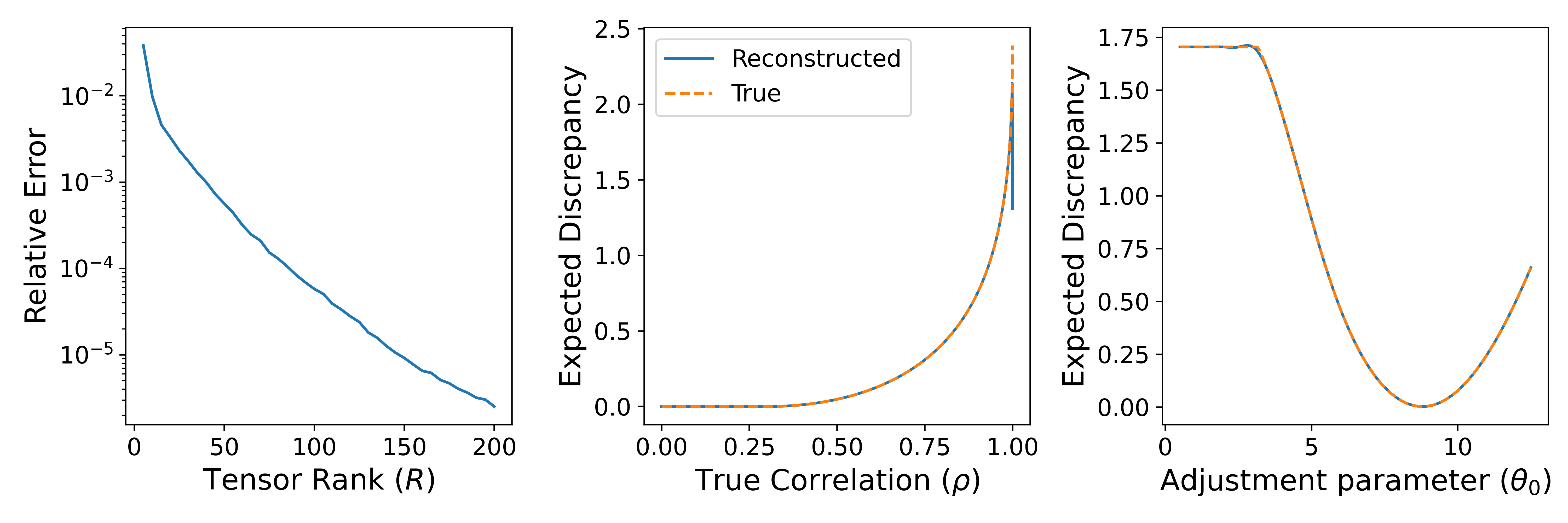}
    \caption{\textbf{(Left)} Relative error of the expected discrepancy array reconstruction as a function of the tensor decomposition rank ($R$). \textbf{(Center)} Expected discrepancy for a fixed $\btheta$ along the true correlation space for both the brute-force computed array (dashed orange) and the reconstruction (solid blue). The reconstructed functions can have slight errors on the boundary of the region. \textbf{(Right)} Expected discrepancy for a fixed $\rho$ and $\theta_1$ shows close agreement between the brute-force and reconstructed versions.}
    \label{fig:exp_discrep_diagnostics}
\end{figure}

\begin{algorithm}[H]
    \caption{Computing the expected discrepancy array}
    \label{alg:exp_discrp_compute}
    \begin{algorithmic}[1] 
    \REQUIRE  An adjustment parameter/correlation space $\Theta \times (0, 1)$, a coarse dimension $d_{\theta_0} \times d_{\theta_1} \times d_\rho$, a final dimension $d'_{\theta_0} \times d'_{\theta_1} \times d'_\rho$, and a tensor rank $R > 0$.
    \ENSURE $\bm{M}_{fine} \in \mathbb{R}^{d'_{\theta_0} \times d'_{\theta_1} \times d'_{\rho}}$.

    \STATE Construct coarse grid over $(\theta_0, \theta_1, \rho) \in \Theta \times (0, 1)$.
    \STATE For $i \in [d_{\theta_0}]$, $j \in [d_{\theta_1}]$, and $k \in [d_{\rho}]$, compute the coarse-grid array $\bm{M}_{coarse} \in \mathbb{R}^{d_{\theta_0} \times d_{\theta_1} \times d_{\rho}}$, where $M_{ijk} = \mathbb{E}_{\hat{\rho}} \left[ \discrepancy \left( g(\btheta_{ij}; \hat{\rho}), \rho_k \right) \right]$. each expectation is compuated by numerically solving \Cref{eq:M_theta_rho}.
    \STATE Perform CP tensor decomposition of array \cite{kossaifi2019tensorly}.
    \STATE Fit a spline to each orthogonal column of decomposed matrices.
    \STATE Up-sample $(\theta_0, \theta_1, \rho)$ grid to desired resolution and use splines to create new up-sampled tensor factors.
    \STATE Reconstruct array at desired resolution.
    \STATE \textbf{return} Fine-grid expected discrepancy array, $\bm{M}_{fine} \in \mathbb{R}^{d'_{\theta_0} \times d'_{\theta_1} \times d'_{\rho}}$.
    \end{algorithmic}
\end{algorithm}

\subsection{Correlation confidence intervals and computing their surrogates} \label{ss:cis_and_surrogates}
Once the expected discrepancy array detailed in \Cref{ss:exp_discp_arr} is computed, the necessary root-finding algorithms to find the confidence interval to solve the DDMM optimization are the computational bottleneck.
This computational burden becomes significant when one wants to solve the DDMM optimization many times (either to look at the adjustment under a collection of $\alpha$ settings or to implement our procedure to choose an optimal $\alpha$ as described in \Cref{ss:optimizing_alpha}).
As shown in \Cref{ss:computing_cis}, we care about equations of the form,
\begin{equation}
    p(\alpha, r; \rho) := \frac{\alpha}{2} - \mathbb{P}\left( \hat{\rho} < r \mid \rho, N \right) = \frac{\alpha}{2} - \int_{-1}^r f(t \mid \rho, N) dt,
\end{equation}
where for a particular $(\alpha, r)$, the root-finding algorithm finds $\rho'$ such that $p(\alpha, r; \rho') = 0$.
The function $p$ is clearly linear in $\alpha$ and smooth in $r$ since the sample correlation density is smooth.
As such, similar to the intuition in \Cref{ss:exp_discp_arr}, the confidence interval endpoints for a particular desired miscoverage level, $\alpha$, and observed sample correlation, $r$ should be similar to the endpoints at a nearby setting, e.g., $(\alpha + \epsilon, r + \epsilon)$, where $\epsilon$ is small.
Using this smoothness intuition, we reduce the computational bottleneck by training confidence interval endpoint surrogate models over the space of possible miscoverage levels and sample correlations (i.e., $(0, 1) \times (-1, 1)$).

To fit these surrogates, we fix a pilot sample size, and define a grid over the miscoverage and sample correlation space.
For each grid point, we solve the lower and upper endpoint root-finding problems.
For each collection of endpoint points (lower and upper), we fit a bivariate spline to generate the endpoint surrogates.
Once the surrogates are fit, we have a computationally efficient way to compute a confidence interval for arbitrary $(\alpha, r)$.
In speed tests performed on a personal computer, the spline computes intervals in the order of microseconds while the root-finding algorithm in the order of milliseconds (a speedup of three orders of magnitude)
This efficiency boost facilitates the optimal-$\alpha$ algorithm in \Cref{ss:optimizing_alpha}.
The full procedure is written in \Cref{alg:ci_surrogates}.

\begin{algorithm}[H]
    \caption{Constructing the confidence interval surrogate models}
    \label{alg:ci_surrogates}
    \begin{algorithmic}[1] 
    \REQUIRE The number of pilot samples $\npilot \in \mathbb{N}$, the number of grid/sample values $n \in \mathbb{N}$
    \ENSURE Lower and upper two-dimensional splines returning correlation confidence intervals for arbitrary $(\alpha, r) \in (0, 1) \times (-1, 1)$ when the sample correlation is composed of $\npilot$ samples. 

    \STATE Grid (or sample) $n$ settings over the space of miscoverage levels and possible observed sample correlations - $(\alpha, r) \in (0, 1) \times (-1, 1)$.
    \STATE For each $\{(\alpha_i, r_i)\}_{i = 1}^n$, compute the $1 - \alpha_i$ confidence interval, $I_i = [l_i, u_i]$, according to the procedure in \Cref{ss:computing_cis}.
    \STATE Using the intervals $\{ I_i \}_{i = 1}^n$, fit a bivariate spline on the lower endpoints and a bivariate spline on the upper endpoints to obtain two endpoint surfaces over $(0, 1) \times (-1, 1)$.
    \end{algorithmic}
\end{algorithm}

\Cref{fig:ci_spline_and_diff} shows the result of the lower and upper endpoint spline fits alongside a comparison of the intervals found via root-finding against those found via the spline.
The bivariate spline well-captures the smooth endpoint surfaces and we observe that the differences between the root-finding and the spline outputs are at worst $1$\% of the total possible confidence interval length.

\begin{figure}
    \centering
    \includegraphics[width=0.99\textwidth]{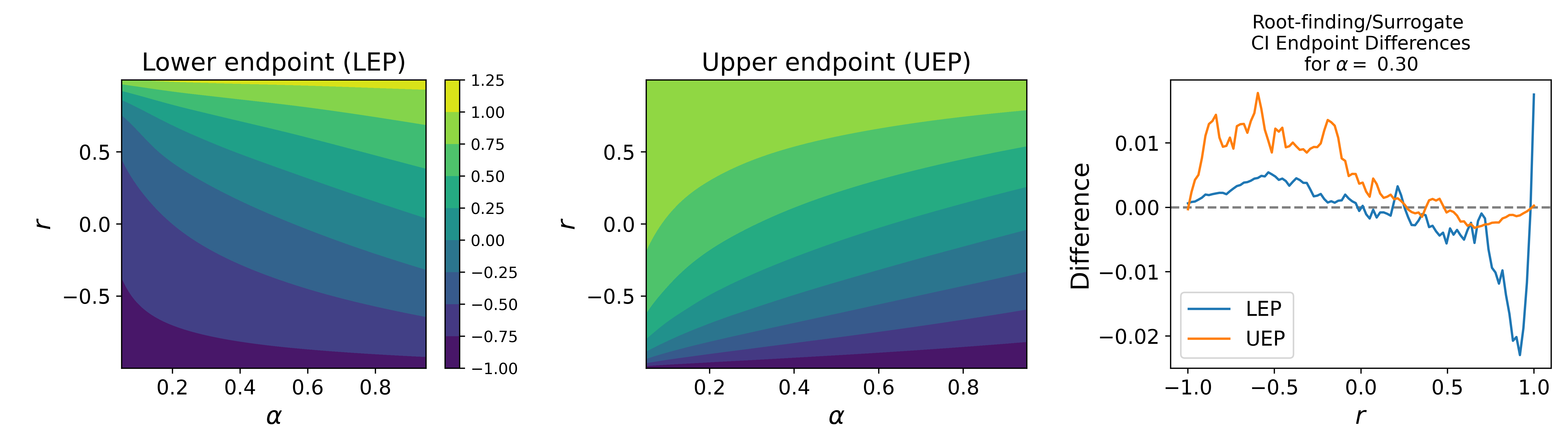}
    \caption{\textbf{(Left)} and \textbf{(Center)} panels show the bivariate spline fits for the lower endpoint (LEP) and upper endpoint (UEP) confidence intervals, respectively. The splines well-capture the smooth endpoint surface over the $(\alpha, r)$ space. \textbf{(Right)} for $\alpha = 0.3$, this panel shows the difference between the interval endpoints found via the root-finding algorithm and the bivariate spline. The absolute difference is less than $0.02$ almost everywhere, so given that the whole range of correlations is length $2$, the worst-case approximation error is approximately $1$\%.}
    \label{fig:ci_spline_and_diff}
\end{figure}

\subsection{Optimizing DDMM confidence level} \label{ss:optimizing_alpha}
This section provides the details of choosing $\alpha$ as described in \Cref{ss:computing_cis}.
For a chosen $\alpha \in (0, 1)$, there is a resulting expected discrepancy difference (EDD, see \Cref{eq:edd}) with respect to the unadjusted sample correlation.
We ultimately wish to make this difference as negative as possible over $\rho \in (0, 1)$ and since $\alpha$ affects the balance between robustness (conservatism) and optimality, we claim that there are $\alpha$ settings achieving this goal.
The following procedure is meant to be performed \emph{before} observing the pilot samples.

Define the following surface over the space of possible miscoverage settings and true correlations, $(\alpha, \rho) \in (\alpha_l, \alpha_u) \times (\rho_l, \rho_u) \subset (0, 1)^2$,
\begin{equation} \label{eq:alpha_rho_surface}
    f(\alpha, \rho) := \mathbb{E}_{\bar{\by}} \left[ \discrepancy(h_\alpha(\bar{\by}), \rho) - \discrepancy(\hat{\rho}, \rho) \right],
\end{equation}
where $h_\alpha(\bar{\by})$ denotes the DDMM estimator defined by $\alpha$ and the $\npilot$ samples, $\bar{\by}$, while $\hat{\rho}$ denotes the usual sample correlation computed via the $\npilot$ samples.
Since we target an estimator $h$ that dominates the sample correlation, we ideally want to pick $\alpha$ such that $f(\alpha, \rho) < 0$ for all $\rho \in (\rho_l, \rho_u)$.
Then, of all the $\alpha$ such that this condition holds, we want to pick the \emph{best} ones, leading to the following minimax quantity,
\begin{equation} \label{eq:alpha_star}
    \alpha^* := \underset{\alpha \in (\alpha_l, \alpha_u)}{\text{argmin}} \underset{\rho \in (\rho_l, \rho_u)}{\max} f(\alpha, \rho),
\end{equation}
i.e., the $\alpha$ minimizing the worst-case expected discrepancy difference.
For each $(\alpha, \rho)$, computing \Cref{eq:alpha_rho_surface} is nontrivial since computing $h_\alpha(\hat{\by})$ involves solving a minimax problem.
However, under the bivariate Gaussian assumption and having computed the expected discrepancy array detailed in \Cref{ss:exp_discp_arr} along with the confidence interval surrogates detailed in \Cref{ss:cis_and_surrogates}, we can easily sample realizations of the random process defined at any $(\alpha, \rho)$.
As such, we can use any regression (ideally nonparametric so we avoid structural assumptions) to estimate the conditional mean surface, i.e., $f(\alpha, \rho)$.
By definition, regression minimizes squared-error loss and provides the optimal estimate of this conditional expectation surface.
Once this surface is estimated, we can numerically solve the minimax problem in \Cref{eq:alpha_star} to obtain the $\alpha$ produces the smallest worst-case expected discrepancy difference.
Clearly, any regression approach to estimate $f(\alpha, \rho)$ will include sampling variability.
Although we do not explicitly characterize this variability for our resulting estimate, we do consider the bias-variance tradeoff resulting from our data-generating process, giving us training data that optimize expected squared error.

\subsubsection{Training data generation, Gaussian Process (GP) regression, and optimizing miscoverage level}
For a fixed $(\alpha, \rho)$, define the following random variable,
\begin{equation}
    \Delta(\alpha, \rho) := \discrepancy(h_\alpha(\bar{\by}), \rho) - \discrepancy(\hat{\rho}, \rho),
\end{equation}
where the randomness arises from the $\npilot$ samples, $\bar{\by}$.
Since $\Delta(\alpha, \rho)$ has a large variance relative to its mean for each $(\alpha, \rho)$, we increase the signal to noise ratio by constructing the training data using $D \in \mathbb{N}$ design points, and $R \in \mathbb{N}$ repetition per design point.
As such, for fixed $(D, R)$, we obtain data of the form,
\begin{equation}
    \mathcal{D} := \left\{(\alpha_i, \rho_i), \left(\hat{\mu}_i, \hat{\sigma}_i^2 \right) \right\}_{i = 1}^D,
\end{equation}
where
\begin{equation} \label{eq:mean_var}
    \hat{\mu}_i = \frac{1}{R} \sum_{j = 1}^R \Delta_{ij}, \quad \hat{\sigma}_i^2 = \frac{1}{R(R - 1)} \sum_{j = 1}^R \left(\Delta_{ij} - \hat{\mu}_i \right)^2, \quad \Delta_{ij} = \Delta(\alpha_i, \rho_i).
\end{equation}
In \Cref{ss:choose_data_sizes} we discuss our approach to allocating training data generating time between design points and repetitions.
Once the number of design points has been chosen, we generate them using Latin Hypercube Sampling (LHS) to ensure more uniform coverage of the desired region compared to sampling uniformly at random.

Given our smoothness assumption on $f(\alpha, \rho)$ and our desire to optimize over the output of the regression, a Gaussian Process regressor is a reasonable model choice \cite{gramacy2020surrogates, rasmussen2006gaussian}.
We use a GP with a nugget term, i.e., we model a deterministic smooth surface that we observe with additive Gaussian noise.
This assumption on the noise is reasonable by the Central Limit Theorem since we use $R$ repetitions to compute a mean for each output value and thus its sampling distribution is asymptotically Gaussian.
As discussed in Chapter 10 of \cite{gramacy2020surrogates}, we use the computed $\hat{\sigma}_i^2$ values to construct a multivariate nugget matrix essentially allowing for heteroskedatic GP regression, allowing the GP regressor to smooth the data and better respect the variance structure across the input space.
Note, the definition includes a $R(R - 1)$ term in the denominator since the standard error of the empirical mean contracts at rate $R^{-1/2}$, i.e., as we increase the number of repetitions, we further sharpen the signal to noise ratio since the sampling variance asymptotically vanishes.

When exploring the GP fitting process, we noticed that the argmax in the $\rho$ space is consistently found in the region where $\rho \geq 0.5$.
As such, we perform LHS on the top half of the desired region as shown in \Cref{fig:gp_surface_and_points}, which further shows the mean surface of the fitted GP.
\Cref{fig:max_over_rho_curve} then shows the pointwise mean surface over $\rho$ for each $\alpha$ which is then used to find $\alpha^*$ (as indicated by the dashed line).
Note that for all $\alpha$ in the computed range, we estimate that DDMM has a worst-case EDD less than zero, i.e., we can choose any $\alpha$ and still out-perform the sample covariance.
But clearly there is a unique $\alpha$ optimizing the performance.

\begin{figure}[htbp]
    \centering
    \begin{subfigure}[t]{0.48\textwidth}
        \centering
        \includegraphics[width=\textwidth]{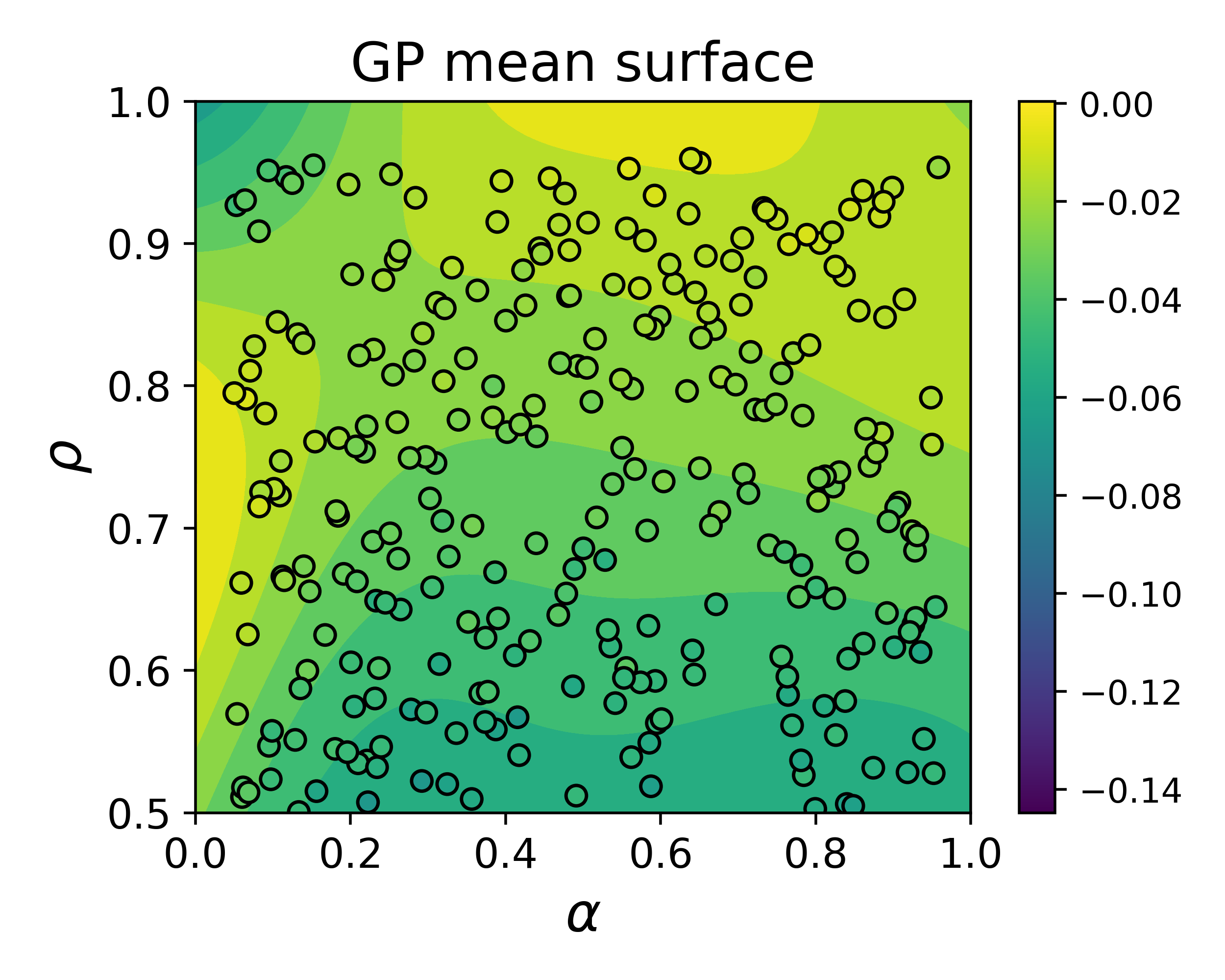}
        \caption{Design points generated from LHS and the mean surface of the fitted Gaussian Process regressor. The mean surface is the optimal estimate for the EDD defined in \Cref{eq:alpha_rho_surface}.}
        \label{fig:gp_surface_and_points}
    \end{subfigure}
    \hfill
    \begin{subfigure}[t]{0.48\textwidth}
        \centering
        \includegraphics[width=\textwidth]{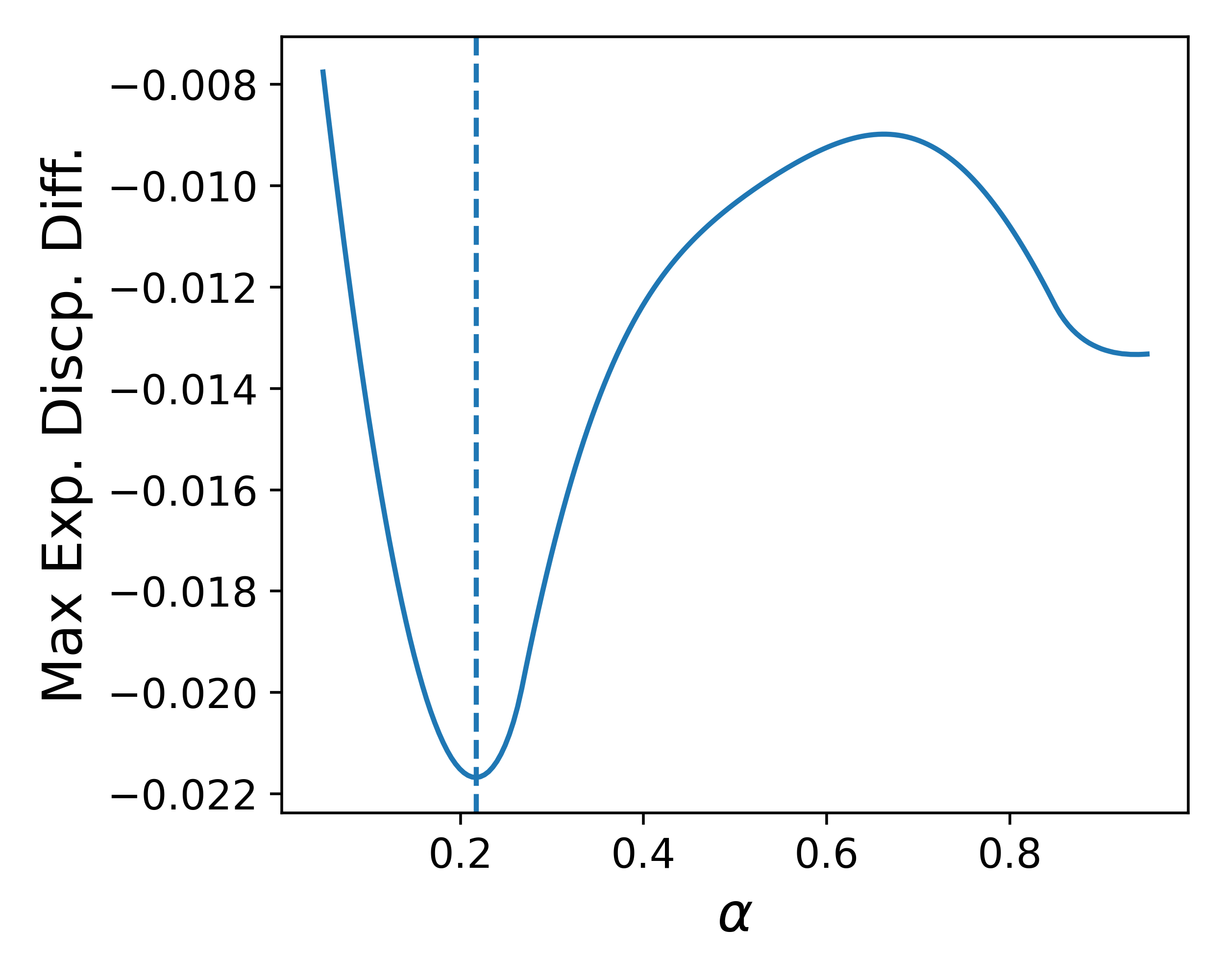}
        \caption{Using the mean GP surface, across the range of $\alpha$ values we show the surface maximum over $\rho$. The location where the resulting curve is minimized (shown via the dashed line) is the estimated optimal miscoverage level, $\alpha^*$.}
        \label{fig:max_over_rho_curve}
    \end{subfigure}
    \caption{Design points, the mean surface of the fitted Gaussian Process regressor, and the pointwise maximum surface for each $\alpha$.}
    \label{fig:gp_fit_and_opt}
\end{figure}

\subsubsection{Choosing number of design points and repetitions} \label{ss:choose_data_sizes}
To constrain the combinations of design points and repetitions to consider, we anchored our set of choices to those combinations that could be generated in $t = 900$ seconds on a personal computer.
All computational experiments were performed on a personal computer running Windows 11, equipped with an Intel Core i5-1245U processor (12 cores, 1.6 GHz) and 316 GB of RAM.
Some experimentation showed that each sample (where one sample is the combined generation of the sample correlation, its confidence interval, and the DDMM solution) requires approximately $r = 0.035$ seconds on the author's laptop.
Thus, we wish to consider all $(D, R)$ such that,
\begin{equation} \label{eq:pareto_front}
    DRr = t.
\end{equation}
Taking the logarithm of both sides yields a linear relationship shown in \Cref{fig:pareto_front}, confirming the intuition that if we use fewer design points, we can obtain more repetitions and vice versa.
To investigate different combinations along this line, we sampled a large dataset of design points ($D = 700$) and repetitions ($R = 276$).
These boundary values are indicated by the gray lines on \Cref{fig:pareto_front} and were chosen somewhat arbitrarily to fit within local computational constraints.
\Cref{fig:pareto_front} shows combinations of $(D, R)$ satisfying \Cref{eq:pareto_front}.
We generate a grid of $(D, R)$ values such that each $D$ value is to the left of the vertical line and each $R$ is below the horizontal line.
We refer to the index enumerating these settings along this Pareto front as the index on the Pareto front, as used in Figures~\eqref{fig:risk_fit} and \eqref{fig:bias_variance}.
Using all design points and repetitions, we estimate the conditional expectation surface in \Cref{eq:alpha_rho_surface} and find the corresponding minimax $\alpha$ as specified by \Cref{eq:alpha_star}.
We refer to this optimal settings at $\alpha_{orcl}$, as it serves the purpose of an oracle setting in choosing the best combination of design points and repetitions.
Note that this setting is depicted as the intersection point of the dashed gray lines in \Cref{fig:pareto_front} and is ``over-budget'' with respect to the Pareto front.

Let $K := (D, R)$ denote an arbitrary design point and repetition setting.
Although all $K$ values falling on the Pareto front have the same computational time, focusing the budget on design point coverage versus repetition count produces different results.
To locate the optimal setting, we propose the following procedure minimizing the \emph{risk} associated with a particular choice, $K$,
\begin{equation} \label{eq:risk_alpha}
    \mathcal{R}(K) := \mathbb{E} \left[\left(\hat{\alpha}_K - \alpha_{orcl} \right)^2 \right] = \text{bias}\left( \hat{\alpha}_K \right)^2 + \text{Var}\left( \hat{\alpha}_K \right),
\end{equation}
where $\hat{\alpha}_K$, denotes the solution to \Cref{eq:alpha_star} using a GP mean surface fitted on $D$ randomly generated design points, each with $R$ randomly generated repetitions and we have included the usual bias-variance decomposition of the risk under the squared loss function.
According to functional linear approximations in empirical process theory \cite{kosorok2008introduction, Pollard1989AsymptoticsVE} on the estimated conditional expectation surface, the bias and variance can be approximated as follows,
\begin{equation}
    \text{bias}\left( \hat{\alpha}_K \right) \approx \frac{b_1}{R}, \quad \text{Var}\left( \hat{\alpha}_K \right) \approx \frac{c_1}{D} + \frac{c_2}{DR}.
\end{equation}
We approximate \Cref{eq:risk_alpha} as follows,
\begin{equation}
    \mathcal{R}(K) \approx \mathcal{R}_{appr}(K; b_1, c_1, c_2) := \left(\frac{b_1}{R} \right)^2 + \frac{c_1}{D} + \frac{c_2}{DR},
\end{equation}
and thus the risk curve along the Pareto front can be approximated by fitting the parameters $(b_1, c_1, c_2)$.

Ideally, for each $K$, we could fully resample the data-generating process by producing $D$ design points and $R$ repetitions to with a large number of resamples to estimate \Cref{eq:risk_alpha}.
Instead, we subsample from our large dataset for each $K$ setting along the Pareto front.
For each setting, $K$, we subsample $D$ design points and $R$ repetitions $100$ times to estimate \Cref{eq:risk_alpha}, yielding a set of pairs,
\begin{equation}
    \left\{\left(K_i, r_i \right) \right\}_{i = 1}^I, \quad r_i = \frac{1}{100} \sum_{j = 1}^{100} (\hat{\alpha}_{K_i} - \alpha_{orcl})^2.
\end{equation}
The pairs $(K_i, r_i)$ are the blue dots in \Cref{fig:risk_fit}.
With these data, we can now solve the following least-squares problem to find the best $(b_1, c_1, c_2)$ settings for $\mathcal{R}_{appr}$,
\begin{align}
    \underset{b_1, c_1, c_2}{\min} &\quad \sum_{i = 1}^{100} \left[\hat{\sigma}_i^{-1} \left(r_i - \mathcal{R}_{appr}(K_i; b_1, c_1, c_2) \right) \right]^2 \nonumber \\
    \text{subject to} &\quad c_1 \geq 0, \; c_2 \geq 0,
\end{align}
where $\hat{\sigma}_i$ is the estimated standard deviation at the $i$th Pareto front index.
The fitted function is shown by the blue solid curve in \Cref{fig:risk_fit} and shows that risk decreases as we move along the Pareto front, i.e., we trade repetitions for design points.
We then investigate the extrapolated risk curve for $K$ settings beyond those accessible from the dataset from which we subsampled, producing the orange solid curve in \Cref{fig:risk_fit}.
We observe that risk again decreases for $K$ values slightly beyond the region of our dataset, but then quickly increases as we continue trading repetitions for design points.
This result makes sense given that we are fundamentally fitting a Gaussian Process regressor from the resulting dataset and that the additive Gaussian noise assumption erodes as we lose repetitions.
Said differently, as we lose repetitions per design point, the regressor model class incurs more systematic bias.
This intuition can be seen by looking at the fitted bias and variance curves in \Cref{fig:bias_variance}.

\begin{figure}[htbp]
    \centering
    \begin{subfigure}[t]{0.32\textwidth}
        \centering
        \includegraphics[width=\textwidth]{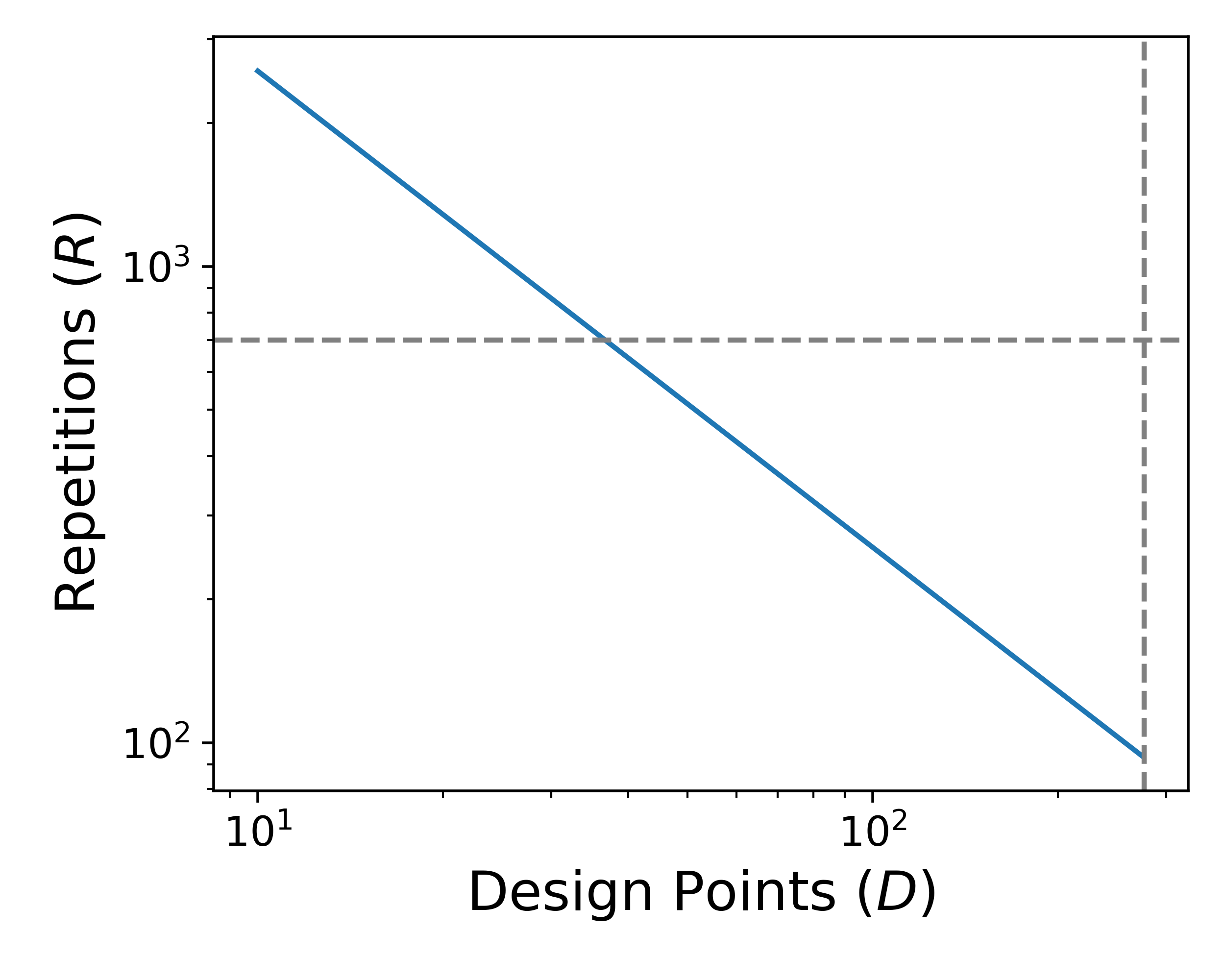}
        \caption{The Pareto front of repetitions ($R$) against design points ($D$) on a log-log scale as defined by \Cref{eq:pareto_front} where the level of the blue line is chosen such that the dataset for computing $\alpha$ can be finished in $t = 900$ seconds on our personal laptop. The dashed gray lines indicate the sizes of the total design point and repetition dataset generated for the purpose of the remaining analysis.}
        \label{fig:pareto_front}
    \end{subfigure}
    \hfill
    \begin{subfigure}[t]{0.32\textwidth}
        \centering
        \includegraphics[width=\textwidth]{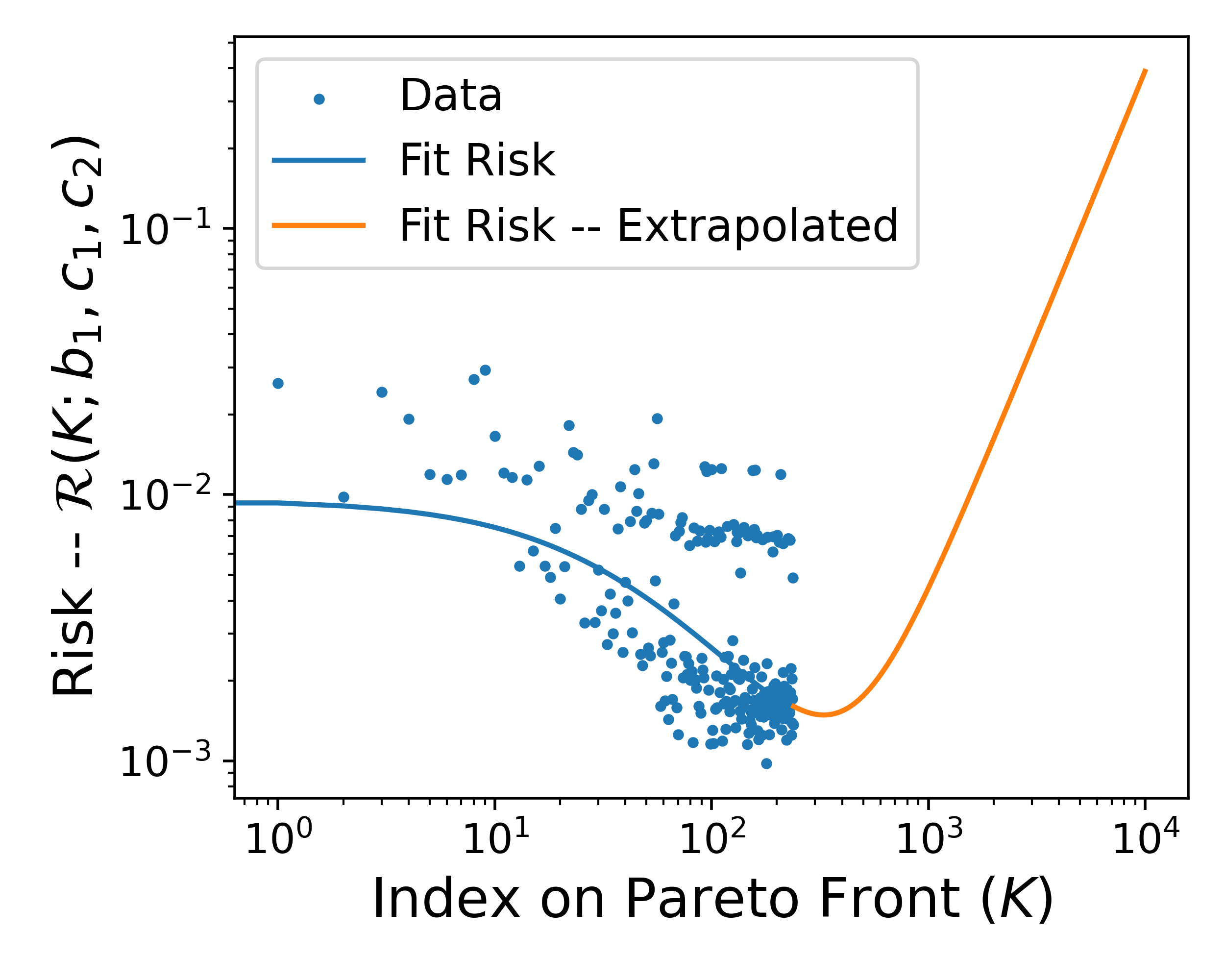}
        \caption{Risk against Pareto front index. The blue dots indicate the empirically computed risk values for different settings ($K_i$) along the Pareto front and the solid lines are the best-fit approximate risk curve.}
        \label{fig:risk_fit}
    \end{subfigure}
    \hfill
    \begin{subfigure}[t]{0.32\textwidth}
        \centering
        \includegraphics[width=\textwidth]{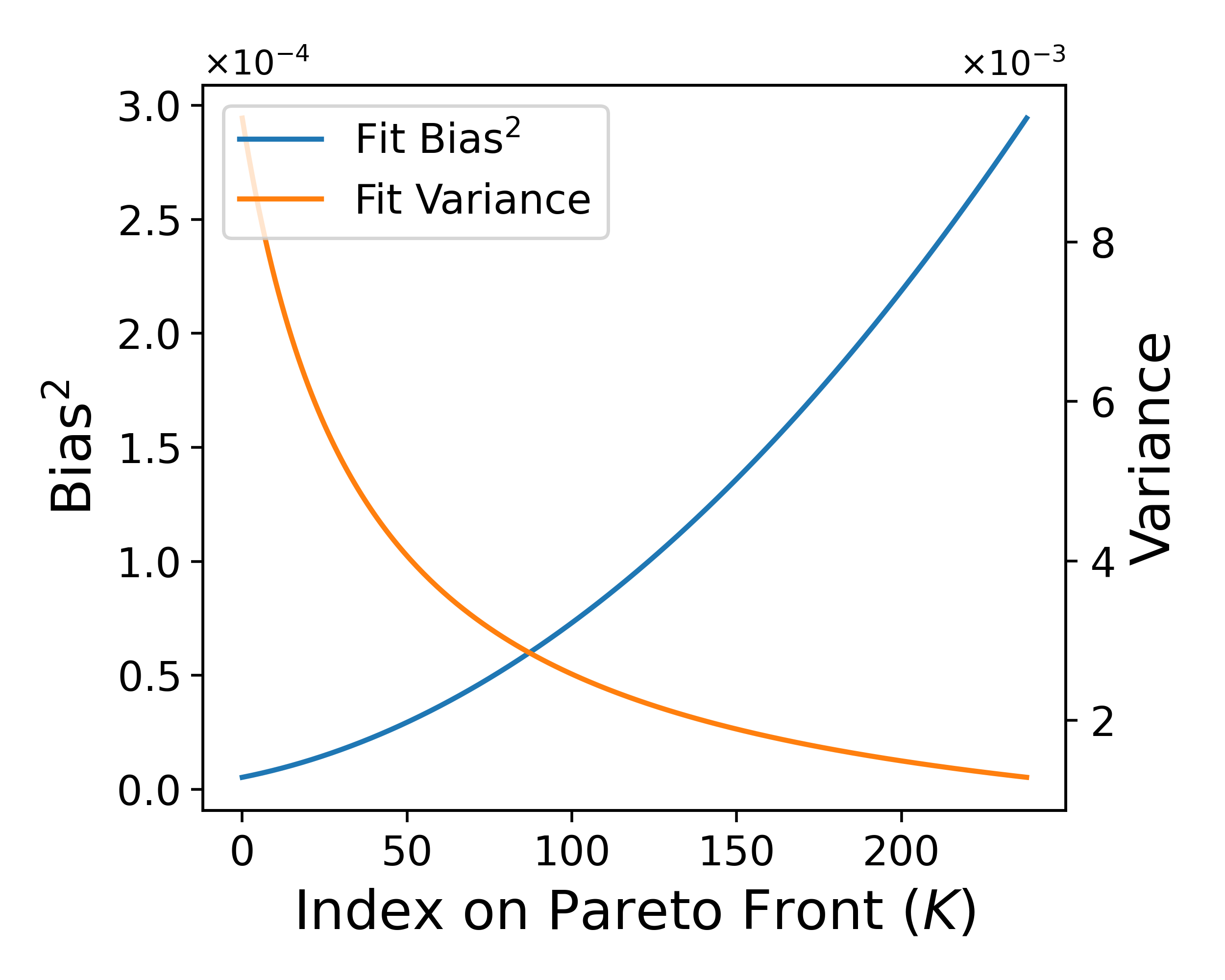}
        \caption{Squared-bias and variance components of the fitted approximate risk function. Bias increases as we move along the Pareto front since lower repetitions means more systematic bias in the GP regression. Variance decreases since more design points means a more rigid mean surface.}
        \label{fig:bias_variance}
    \end{subfigure}
    \caption{Results for determining the optimal $K = (D, R)$ setting along the Pareto front, where the optimal $K$ is the setting minimizing the expected square loss of the optimized $\hat{\alpha}_K$ with respect to $\alpha_{orcl}$.}
    \label{fig:design_rep_images}
\end{figure}

Overall, this analysis supports a rule of thumb to allocate computational resources for choosing $\alpha$ more to design points than repetitions, while keeping in mind that the number of repetitions should not be too small so as to avoid substantial systematic bias in the Gaussian Process regression.

\begin{algorithm}[H]
    \caption{Constructing the confidence interval surrogate models}
    \label{alg:opt_alpha}
    \begin{algorithmic}[1] 
    \REQUIRE Number of design points $D \in \mathbb{N}$ and number of repetitions per design point $R \in \mathbb{R}$.
    \ENSURE $\alpha^*$ minimizing the estimated worst-case expected discrepancy difference.

    \STATE Use Latin Hypercube sampling to generate $D$ values in $(\alpha, \rho) \in (\alpha_l, \alpha_u) \times (\rho_l, \rho_u)$.
    \FOR{$i = 1$ to $D$}
        \FOR{$j = 1$ to $R$}
            \STATE Sample $\hat{\rho}_{ij}$ under its sampling distribution when $\rho_i$ is the true correlation.
            \STATE Use DDMM to compute the adjusted sample correlation: $\hat{\rho}^{ij}_a = g(\hat{\btheta}_{ij}; \hat{\rho}_{ij})$, where $\hat{\btheta}_{ij}$ is the DDMM adjustment parameter for the $i$th design point and $j$th repetition.
            \STATE Compute the discrepancy difference value: $\Delta_{ij} := \discrepancy(\hat{\rho}^{ij}_a, \rho_i) - \discrepancy(\hat{\rho}_{ij}, \rho_i)$.
        \ENDFOR
    \ENDFOR
    \STATE Obtain mean and variance values at each design point using \Cref{eq:mean_var}.
    \STATE Fit a GP to estimate the conditional mean surface \Cref{eq:alpha_rho_surface}, $\hat{f}(\alpha, \rho)$. This surface estimates the expected discrepancy difference for $(\alpha, \rho)$ in the given hyperrectangle.
    \STATE Use $\hat{f}$ to compute the $\alpha$ minimizing the worst-case expected discrepancy difference:
    \begin{equation}
        \alpha^* = \underset{\alpha \in (\alpha_l, \alpha_u)}{\text{argmin}} \underset{\rho \in (\rho_l, \rho_u)}{\max} \; \hat{f}(\alpha, \rho).
    \end{equation}
    \STATE \textbf{return} Optimized $\alpha^*$.
    \end{algorithmic}
\end{algorithm}

\bibliographystyle{apalike}
\bibliography{references}

\end{document}